\documentclass[sigconf]{acmart}

\AtBeginDocument{%
  \providecommand\BibTeX{{%
    \normalfont B\kern-0.5em{\scshape i\kern-0.25em b}\kern-0.8em\TeX}}}

\usepackage{algorithm,algorithmicx}
\usepackage[noend]{algpseudocode}
\usepackage{setspace}

\algnewcommand{\LeftComment}[1]{\Statex \(\triangleright\) #1}

\usepackage{xspace}

\usepackage{thmtools}
\usepackage{thm-restate}

\usepackage{subcaption}

\newcommand{\mcA}{\mathcal A}

\newcommand{\mcT}{\mathcal T}

\newcommand{\RR}{\mathbb{R}}

\newcommand{\DOTTT}{{\tt DOTTT}}

\newcommand{\Sec}[1]{\hyperref[sec:#1]{\S\ref*{sec:#1}}} 
\newcommand{\Apx}[1]{\hyperref[apx:#1]{Appendix~\ref*{apx:#1}}} 
\newcommand{\Eqn}[1]{\hyperref[eq:#1]{(\ref*{eq:#1})}} 
\newcommand{\Fig}[1]{\hyperref[fig:#1]{Fig.\,\ref*{fig:#1}}} 
\newcommand{\Tab}[1]{\hyperref[tab:#1]{Tab.\,\ref*{tab:#1}}} 
\newcommand{\Thm}[1]{\hyperref[thm:#1]{Theorem\,\ref*{thm:#1}}} 
\newcommand{\Fact}[1]{\hyperref[fact:#1]{Fact\,\ref*{fact:#1}}} 
\newcommand{\Lem}[1]{\hyperref[lem:#1]{Lemma\,\ref*{lem:#1}}} 
\newcommand{\Prop}[1]{\hyperref[prop:#1]{Prop.~\ref*{prop:#1}}} 
\newcommand{\Cor}[1]{\hyperref[cor:#1]{Corollary~\ref*{cor:#1}}} 
\newcommand{\Conj}[1]{\hyperref[conj:#1]{Conjecture~\ref*{conj:#1}}} 
\newcommand{\Def}[1]{\hyperref[def:#1]{Definition~\ref*{def:#1}}} 
\newcommand{\Alg}[1]{\hyperref[alg:#1]{Alg.~\ref*{alg:#1}}} 
\newcommand{\Ex}[1]{\hyperref[ex:#1]{Ex.~\ref*{ex:#1}}} 
\newcommand{\Clm}[1]{\hyperref[clm:#1]{Claim~\ref*{clm:#1}}} 
\newcommand{\Obs}[1]{\hyperref[obs:#1]{Obs.~\ref*{obs:#1}}} 

\newcommand{\etal}{{et al.}\xspace}

\newcommand{\first}{f}
\newcommand{\last}{\ell}

\newcommand{\degen}{\kappa}

\newcommand{\tempOrder}{\pi}
\newcommand{\dirType}{\rho}
\newcommand{\TriType}{\mcT}
\newcommand{\TriTypeConversion}{\psi}
\newcommand{\numStaticTriangles}{\tau}
\newcommand{\threeDelta}{(\delta_{1,3},\delta_{1,2},\delta_{2,3})}

\DeclareMathOperator{\consecutivetriangleCount}{TTC}
\DeclareMathOperator{\edgeCount}{EC}
\DeclareMathOperator{\cumulativeEdgeCount}{CEC}

\newcommand{\PBL}{PBL\space}

\usepackage{tikz}
\usetikzlibrary{decorations.markings,fit,arrows,positioning,backgrounds}
\tikzset{%
  gnode/.style={shape=circle,minimum size=3mm,fill,draw=black}
}
\tikzset{myptr/.style={decoration={markings,mark=at position 1 with {\arrow[scale=2,>=stealth]{>}}},postaction={decorate}}}

\def\OrderOne{
    \node (1) [label={below left:\small u}] at (0,0) [nd, fill=blue] {};
	\node (2) [label={below right:\small v}] at (1,0) [nd, fill=blue] {};
	\node (3) [label={above:\small w}] at (0.5,1) [nd, fill=blue] {};

	\draw (1) to (2);
	\draw (1) to (3);
	\draw (2) to (3);
	
	\node [inner sep = 0,left] at +(0.2,0.65) {\small 2};
	\node [inner sep = 0,right] at +(0.8,0.65) {\small 3};
	\node [inner sep = 0,below] at +(0.5,-0.1) {\small 1};
	
	\node [inner sep = 0,below] at +(0.5,-0.6) {{\Large $\tempOrder_1$}};
}

\def\OrderTwo{
    \node (1) [label={below left:\small u}] at (0,0) [nd, fill=blue] {};
	\node (2) [label={below right:\small v}] at (1,0) [nd, fill=blue] {};
	\node (3) [label={above:\small w}] at (0.5,1) [nd, fill=blue] {};

	\draw (1) to (2);
	\draw (1) to (3);
	\draw (2) to (3);
	
	\node [inner sep = 0,left] at +(0.2,0.65) {\small 1};
	\node [inner sep = 0,right] at +(0.8,0.65) {\small 3};
	\node [inner sep = 0,below] at +(0.5,-0.1) {\small 2};
	
	\node [inner sep = 0,below] at +(0.5,-0.6) {{\Large $\tempOrder_2$}};
}

\def\OrderThree{
    \node (1) [label={below left:\small u}] at (0,0) [nd, fill=blue] {};
	\node (2) [label={below right:\small v}] at (1,0) [nd, fill=blue] {};
	\node (3) [label={above:\small w}] at (0.5,1) [nd, fill=blue] {};

	\draw (1) to (2);
	\draw (1) to (3);
	\draw (2) to (3);
	
	\node [inner sep = 0,left] at +(0.2,0.65) {\small 3};
	\node [inner sep = 0,right] at +(0.8,0.65) {\small 2};
	\node [inner sep = 0,below] at +(0.5,-0.1) {\small 1};
	
	\node [inner sep = 0,below] at +(0.5,-0.6) {{\Large $\tempOrder_3$}};
}

\def\OrderFour{
    \node (1) [label={below left:\small u}] at (0,0) [nd, fill=blue] {};
	\node (2) [label={below right:\small v}] at (1,0) [nd, fill=blue] {};
	\node (3) [label={above:\small w}] at (0.5,1) [nd, fill=blue] {};

	\draw (1) to (2);
	\draw (1) to (3);
	\draw (2) to (3);
	
	\node [inner sep = 0,left] at +(0.2,0.65) {\small 1};
	\node [inner sep = 0,right] at +(0.8,0.65) {\small 2};
	\node [inner sep = 0,below] at +(0.5,-0.1) {\small 3};
	
	\node [inner sep = 0,below] at +(0.5,-0.6) {{\Large $\tempOrder_4$}};
}

\def\OrderFive{
    \node (1) [label={below left:\small u}] at (0,0) [nd, fill=blue] {};
	\node (2) [label={below right:\small v}] at (1,0) [nd, fill=blue] {};
	\node (3) [label={above:\small w}] at (0.5,1) [nd, fill=blue] {};

	\draw (1) to (2);
	\draw (1) to (3);
	\draw (2) to (3);
	
	\node [inner sep = 0,left] at +(0.2,0.65) {\small 3};
	\node [inner sep = 0,right] at +(0.8,0.65) {\small 1};
	\node [inner sep = 0,below] at +(0.5,-0.1) {\small 2};
	
	\node [inner sep = 0,below] at +(0.5,-0.6) {{\Large $\tempOrder_5$}};
}

\def\OrderSix{
    \node (1) [label={below left:\small u}] at (0,0) [nd, fill=blue] {};
	\node (2) [label={below right:\small v}] at (1,0) [nd, fill=blue] {};
	\node (3) [label={above:\small w}] at (0.5,1) [nd, fill=blue] {};

	\draw (1) to (2);
	\draw (1) to (3);
	\draw (2) to (3);
	
	\node [inner sep = 0,left] at +(0.2,0.65) {\small 2};
	\node [inner sep = 0,right] at +(0.8,0.65) {\small 1};
	\node [inner sep = 0,below] at +(0.5,-0.1) {\small 3};
	
	\node [inner sep = 0,below] at +(0.5,-0.6) {{\Large $\tempOrder_6$}};
}

\def\DirOne{
    \node (1) [label={below left:\small u}] at (0,0) [nd, fill=blue] {};
	\node (2) [label={below right:\small v}] at (1,0) [nd, fill=blue] {};
	\node (3) [label={above:\small w}] at (0.5,1) [nd, fill=blue] {};

	\draw[myptr] (1) to (2);
	\draw[myptr] (1) to (3);
	\draw[myptr] (2) to (3);
	
	\node [inner sep = 0,below] at +(0.5,-0.6) {{\Large $\dirType_1$}};
}

\def\DirTwo{
    \node (1) [label={below left:\small u}] at (0,0) [nd, fill=blue] {};
	\node (2) [label={below right:\small v}] at (1,0) [nd, fill=blue] {};
	\node (3) [label={above:\small w}] at (0.5,1) [nd, fill=blue] {};

	\draw[myptr] (1) to (2);
	\draw[myptr] (3) to (1);
	\draw[myptr] (3) to (2);
	
	\node [inner sep = 0,below] at +(0.5,-0.6) {{\Large $\dirType_2$}};
}

\def\DirThree{
    \node (1) [label={below left:\small u}] at (0,0) [nd, fill=blue] {};
	\node (2) [label={below right:\small v}] at (1,0) [nd, fill=blue] {};
	\node (3) [label={above:\small w}] at (0.5,1) [nd, fill=blue] {};

	\draw[myptr] (1) to (2);
	\draw[myptr] (1) to (3);
	\draw[myptr] (3) to (2);
	
	\node [inner sep = 0,below] at +(0.5,-0.6) {{\Large $\dirType_3$}};
}

\def\DirFour{
    \node (1) [label={below left:\small u}] at (0,0) [nd, fill=blue] {};
	\node (2) [label={below right:\small v}] at (1,0) [nd, fill=blue] {};
	\node (3) [label={above:\small w}] at (0.5,1) [nd, fill=blue] {};

	\draw[myptr] (1) to (2);
	\draw[myptr] (3) to (1);
	\draw[myptr] (2) to (3);
	
	\node [inner sep = 0,below] at +(0.5,-0.6) {{\Large $\dirType_4$}};
}

\def\DirFive{
    \node (1) [label={below left:\small u}] at (0,0) [nd, fill=blue] {};
	\node (2) [label={below right:\small v}] at (1,0) [nd, fill=blue] {};
	\node (3) [label={above:\small w}] at (0.5,1) [nd, fill=blue] {};

	\draw[myptr] (2) to (1);
	\draw[myptr] (1) to (3);
	\draw[myptr] (2) to (3);
	
	\node [inner sep = 0,below] at +(0.5,-0.6) {{\Large $\dirType_5$}};
}

\def\DirSix{
    \node (1) [label={below left:\small u}] at (0,0) [nd, fill=blue] {};
	\node (2) [label={below right:\small v}] at (1,0) [nd, fill=blue] {};
	\node (3) [label={above:\small w}] at (0.5,1) [nd, fill=blue] {};

	\draw[myptr] (2) to (1);
	\draw[myptr] (3) to (1);
	\draw[myptr] (3) to (2);
	
	\node [inner sep = 0,below] at +(0.5,-0.6) {{\Large $\dirType_6$}};
}

\def\DirSeven{
    \node (1) [label={below left:\small u}] at (0,0) [nd, fill=blue] {};
	\node (2) [label={below right:\small v}] at (1,0) [nd, fill=blue] {};
	\node (3) [label={above:\small w}] at (0.5,1) [nd, fill=blue] {};

	\draw[myptr] (2) to (1);
	\draw[myptr] (1) to (3);
	\draw[myptr] (3) to (2);
	
	\node [inner sep = 0,below] at +(0.5,-0.6) {{\Large $\dirType_7$}};
}

\def\DirEight{
    \node (1) [label={below left:\small u}] at (0,0) [nd, fill=blue] {};
	\node (2) [label={below right:\small v}] at (1,0) [nd, fill=blue] {};
	\node (3) [label={above:\small w}] at (0.5,1) [nd, fill=blue] {};

	\draw[myptr] (2) to (1);
	\draw[myptr] (3) to (1);
	\draw[myptr] (2) to (3);
	
	\node [inner sep = 0,below] at +(0.5,-0.6) {{\Large $\dirType_8$}};
}

\def\TriTypeOne{
    \node (1) at (0,0) [nd, fill=red] {};
	\node (2) at (1,0) [nd, fill=green] {};
	\node (3) at (0.5,1) [nd, fill=blue] {};

	\draw[myptr] (1) to (2);
	\draw[myptr] (1) to (3);
	\draw[myptr] (3) to (2);

	\node [inner sep = 0,left] at +(0.2,0.65) {\small 3};
	\node [inner sep = 0,right] at +(0.8,0.65) {\small 2};
	\node [inner sep = 0,below] at +(0.5,-0.1) {\small 1};
	
	\node [inner sep = 0,below] at +(0.5,-0.6) {{\Large $\TriType_1$}};
}

\def\TriTypeTwo{
    \node (1) at (0,0) [nd, fill=red] {};
	\node (2) at (1,0) [nd, fill=green] {};
	\node (3) at (0.5,1) [nd, fill=blue] {};

	\draw[myptr] (1) to (2);
	\draw[myptr] (3) to (1);
	\draw[myptr] (3) to (2);

	\node [inner sep = 0,left] at +(0.2,0.65) {\small 3};
	\node [inner sep = 0,right] at +(0.8,0.65) {\small 2};
	\node [inner sep = 0,below] at +(0.5,-0.1) {\small 1};
	
	\node [inner sep = 0,below] at +(0.5,-0.6) {{\Large $\TriType_2$}};
}

\def\TriTypeThree{
    \node (1) at (0,0) [nd, fill=red] {};
	\node (2) at (1,0) [nd, fill=green] {};
	\node (3) at (0.5,1) [nd, fill=blue] {};

	\draw[myptr] (1) to (2);
	\draw[myptr] (1) to (3);
	\draw[myptr] (2) to (3);

	\node [inner sep = 0,left] at +(0.2,0.65) {\small 3};
	\node [inner sep = 0,right] at +(0.8,0.65) {\small 2};
	\node [inner sep = 0,below] at +(0.5,-0.1) {\small 1};
	
	\node [inner sep = 0,below] at +(0.5,-0.6) {{\Large $\TriType_3$}};
}

\def\TriTypeFour{
    \node (1) at (0,0) [nd, fill=red] {};
	\node (2) at (1,0) [nd, fill=green] {};
	\node (3) at (0.5,1) [nd, fill=blue] {};

	\draw[myptr] (1) to (2);
	\draw[myptr] (3) to (1);
	\draw[myptr] (2) to (3);

	\node [inner sep = 0,left] at +(0.2,0.65) {\small 3};
	\node [inner sep = 0,right] at +(0.8,0.65) {\small 2};
	\node [inner sep = 0,below] at +(0.5,-0.1) {\small 1};
	
	\node [inner sep = 0,below] at +(0.5,-0.6) {{\Large $\TriType_4$}};
}

\def\TriTypeFive{
    \node (1) at (0,0) [nd, fill=red] {};
	\node (2) at (1,0) [nd, fill=green] {};
	\node (3) at (0.5,1) [nd, fill=blue] {};

	\draw[myptr] (1) to (2);
	\draw[myptr] (1) to (3);
	\draw[myptr] (3) to (2);

	\node [inner sep = 0,left] at +(0.2,0.65) {\small 2};
	\node [inner sep = 0,right] at +(0.8,0.65) {\small 3};
	\node [inner sep = 0,below] at +(0.5,-0.1) {\small 1};
	
	\node [inner sep = 0,below] at +(0.5,-0.6) {{\Large $\TriType_4$}};
}

\def\TriTypeSix{
    \node (1) at (0,0) [nd, fill=red] {};
	\node (2) at (1,0) [nd, fill=green] {};
	\node (3) at (0.5,1) [nd, fill=blue] {};

	\draw[myptr] (1) to (2);
	\draw[myptr] (3) to (1);
	\draw[myptr] (3) to (2);

	\node [inner sep = 0,left] at +(0.2,0.65) {\small 2};
	\node [inner sep = 0,right] at +(0.8,0.65) {\small 3};
	\node [inner sep = 0,below] at +(0.5,-0.1) {\small 1};
	
	\node [inner sep = 0,below] at +(0.5,-0.6) {{\Large $\TriType_6$}};
}

\def\TriTypeSeven{
    \node (1) at (0,0) [nd, fill=red] {};
	\node (2) at (1,0) [nd, fill=green] {};
	\node (3) at (0.5,1) [nd, fill=blue] {};

	\draw[myptr] (1) to (2);
	\draw[myptr] (1) to (3);
	\draw[myptr] (2) to (3);

	\node [inner sep = 0,left] at +(0.2,0.65) {\small 2};
	\node [inner sep = 0,right] at +(0.8,0.65) {\small 3};
	\node [inner sep = 0,below] at +(0.5,-0.1) {\small 1};
	
	\node [inner sep = 0,below] at +(0.5,-0.6) {{\Large $\TriType_7$}};
}

\def\TriTypeEight{
    \node (1) at (0,0) [nd, fill=red] {};
	\node (2) at (1,0) [nd, fill=green] {};
	\node (3) at (0.5,1) [nd, fill=blue] {};

	\draw[myptr] (1) to (2);
	\draw[myptr] (3) to (1);
	\draw[myptr] (2) to (3);

	\node [inner sep = 0,left] at +(0.2,0.65) {\small 2};
	\node [inner sep = 0,right] at +(0.8,0.65) {\small 3};
	\node [inner sep = 0,below] at +(0.5,-0.1) {\small 1};
	
	\node [inner sep = 0,below] at +(0.5,-0.6) {{\Large $\TriType_8$}};
}

\copyrightyear{2021}
\acmYear{2021}
\setcopyright{acmcopyright}\acmConference[KDD '21]{Proceedings of the 27th ACM SIGKDD Conference on Knowledge Discovery and Data Mining}{August 14--18, 2021}{Virtual Event, Singapore}
\acmBooktitle{Proceedings of the 27th ACM SIGKDD Conference on Knowledge Discovery and Data Mining (KDD '21), August 14--18, 2021, Virtual Event, Singapore}
\acmPrice{15.00}
\acmDOI{10.1145/3447548.3467374}
\acmISBN{978-1-4503-8332-5/21/08}

\begin{document}

\fancyhead{}


\settopmatter{printfolios=true}

\title{Faster and Generalized Temporal Triangle Counting, via Degeneracy Ordering}


\author{Noujan Pashanasangi}
\email{npashana@ucsc.edu}
\affiliation{%
\institution{University of California, Santa Cruz}
 \city{Santa Cruz}
 \state{California}
 \country{USA}}

\author{C. Seshadhri}
\email{sesh@ucsc.edu}
\affiliation{%
 \institution{University of California, Santa Cruz}
 \city{Santa Cruz}
 \state{California}
 \country{USA}}



\begin{abstract}
Triangle counting is a fundamental technique in network analysis, that has received much attention in various input models. The vast majority of triangle counting algorithms are targeted to static graphs. Yet, many real-world graphs are directed and \emph{temporal}, where edges come with timestamps.
Temporal triangles yield much more information, since they account for both the graph topology and the timestamps.

Temporal triangle counting has seen a few recent results, but there are varying definitions of temporal triangles. In all cases, temporal triangle patterns enforce constraints on the time interval between edges (in the triangle).
We define a general notion $(\delta_{1,3}, \delta_{1,2}, \delta_{2,3})$-temporal triangles that allows for separate time constraints for all pairs of edges.

Our main result is a new algorithm, \DOTTT{} (Degeneracy Oriented Temporal Triangle Totaler), that exactly counts all directed variants of $(\delta_{1,3}, \delta_{1,2}, \delta_{2,3})$-temporal triangles. Using the classic idea of degeneracy ordering with careful combinatorial arguments, we can prove that \DOTTT{} runs in $O(m\degen\log m)$ time, where $m$ is the number of (temporal) edges and $\degen$ is the graph degeneracy (max core number). Up to log factors, this matches the running time of the best static triangle counters. Moreover, this running time is better than existing.

\DOTTT{} has excellent practical behavior and runs twice as fast as existing state-of-the-art temporal triangle counters (and is also more general). For example, \DOTTT{} computes all types of temporal queries in Bitcoin temporal network with half a billion edges in less than an hour on a commodity machine.
\end{abstract}

\ccsdesc[500]{Theory of computation~Graph algorithms analysis}
\ccsdesc[300]{Information systems~Data mining}
\ccsdesc[500]{Mathematics of computing~Graph algorithms}
\ccsdesc[300]{Information systems~Social networks}
\keywords{triangle counting; temporal networks; degeneracy}

\maketitle

\section{Introduction}
Triangle counting is a fundamental problem in network analysis. There has been a rich line of work on counting triangles in graphs~\cite{itai1978finding,alon1997finding,latapy2008main,seshadhri2013triadic,al2018triangle}. The triangle counts appear in form of different parameters such as \emph{clustering coefficient}~\cite{watts1998collective} and \emph{transitivity ratios}~\cite{wasserman1994social}. Triangle counting has many applications such as social networks analysis~\cite{pfeiffer2012fast}, indexing graph databases~\cite{khan2011neighborhood}, community discovery~\cite{palla2005uncovering}, and spam detection~\cite{becchetti2008efficient}. 


Much of the rich history of triangle counting has focused on {static} graphs. But many real-world networks, such as communication networks, message networks, and social interaction networks are fundamentally \emph{temporal}. Every edge has an associated timestamp~\cite{kovanen2011temporal, gaumont2016finding, farajtabar2018coevolve}. We can model these attributed networks as \emph{temporal networks} where edges have timestamp. For example, in cryptocurrency transaction networks and email networks, each link is between a sender and a receiver and has a timestamp that could be represented as a directed edge with a timestamp in a temporal network. 

\emph{Temporal triangle counts} provide a far richer set of counts than standard counts. These counts take into account the temporal ordering of edges in a triangle, and potentially impose constraints on the timestamp difference among edges. Temporal triangle and motif counting has applications in graph representation learning~\cite{tu2019gl2vec}, expressivity of graph neural networks (GNNs)~\cite{bouritsas2020improving}, network classification~\cite{tu2018network}, temporal text network analysis~\cite{vega2018foundations}, computer networks~\cite{valverde2005network}, and brain networks~\cite{dimitriadis2010tracking}.
Recently, there has been significant interest in temporal triangle and motif counting algorithms~\cite{mackey2018chronological,liu2018sampling,tu2018network,petrovic2019counting,sun2019new,tu2019gl2vec,boekhout2019efficiently,bouritsas2020improving,wang2020efficient}.

Counting temporal triangles in (directed) temporal networks introduces new challenges to that of triangle counting in static graphs. The first challenge is actually defining types of temporal triangles (or motifs). In essence, all definitions specify constraints on the time difference between edges of a triangle.
For example, Kovanen~\etal~\cite{kovanen2011temporal} restrict temporal triangles to cases where the gap between two consecutive edges in the temporal ordering is at most $\Delta$ time unit, and the two edges incident to each node are consecutive event of that node.
Paranjape-Benson-Leskovec (henceforth PBL) introduce \emph{$\delta$-temporal triangles}, where all edges of the triangle/motif have to occur within $\delta$ timesteps~\cite{paranjape2017motifs}. These varying definitions necessitate different algorithms. Our first motivating question is whether one can design algorithms for a more general notion of temporal triangles.

Secondly, there is a significant gap between the best static triangle counting algorithms and temporal triangle counters. Specifically, the classic and immensely practical triangle counting algorithm of Chiba-Nishizeki runs in time $O(m\degen)$, where $m$ is the number of edges and $\degen$ is the \emph{graph degeneracy} (or max core number). The current state-of-the-art temporal triangle counting algorithm of PBL runs in $O(m\sqrt{\numStaticTriangles})$ time, where $\numStaticTriangles$ is the total triangle count (of the underlying static graph). There is a large gap between $\degen$ (which is typically in the hundreds and thought of as a constant) and $\numStaticTriangles$ (which is superlinear in $m$).

These twin issues motivate our study. \emph{Can we define a more general notion of temporal triangles, and give an algorithm whose asymptotic running time is closer to that of static triangle counting?}

\begin{figure*}[th]
\centering
\begin{subfigure}[b]{0.33\textwidth}
\centering
\includegraphics[width=\textwidth]{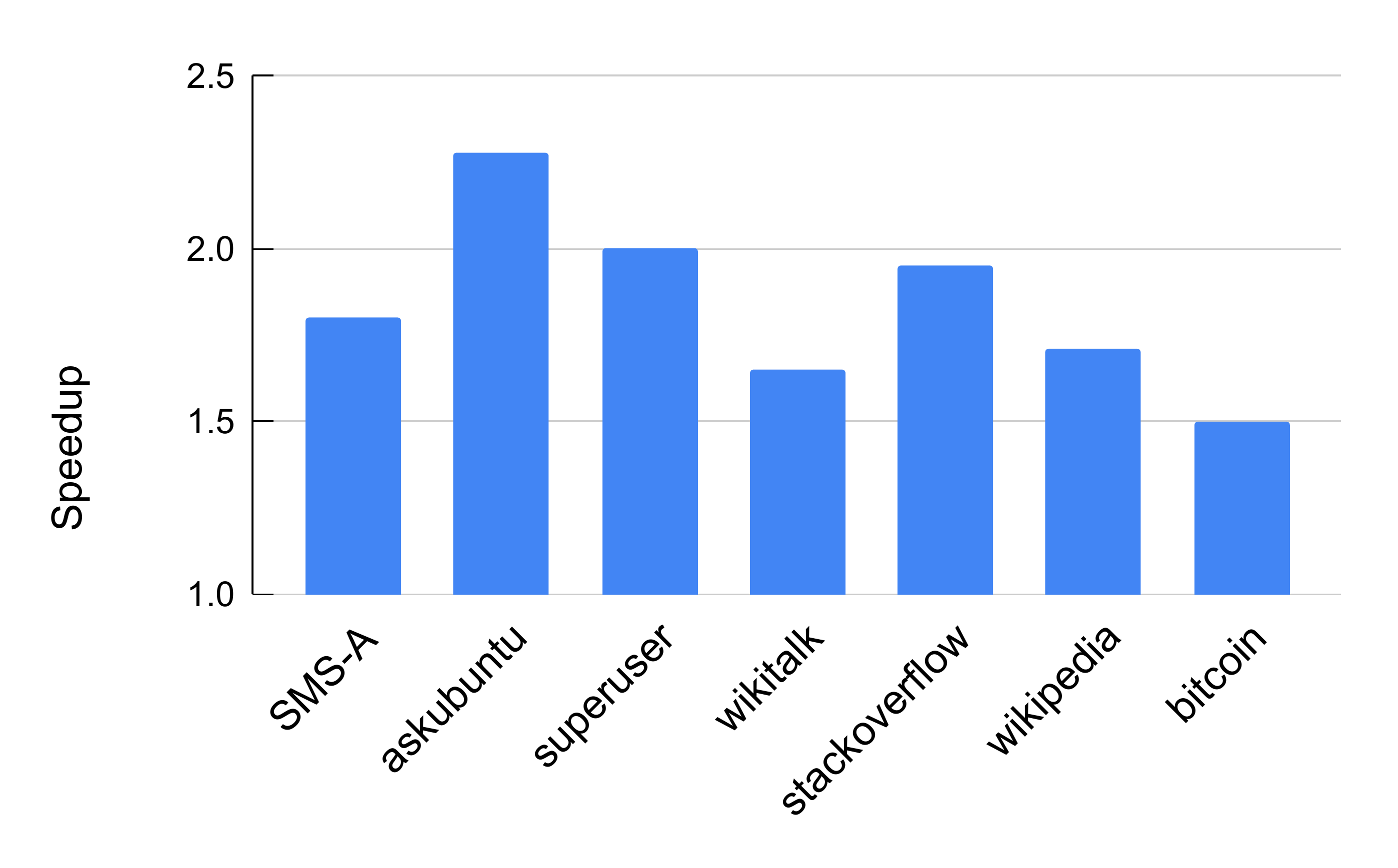}
\caption{Speedup}
\label{fig:speedup}
\end{subfigure}
\begin{subfigure}[b]{0.33\textwidth}
\centering
\includegraphics[width=\columnwidth]{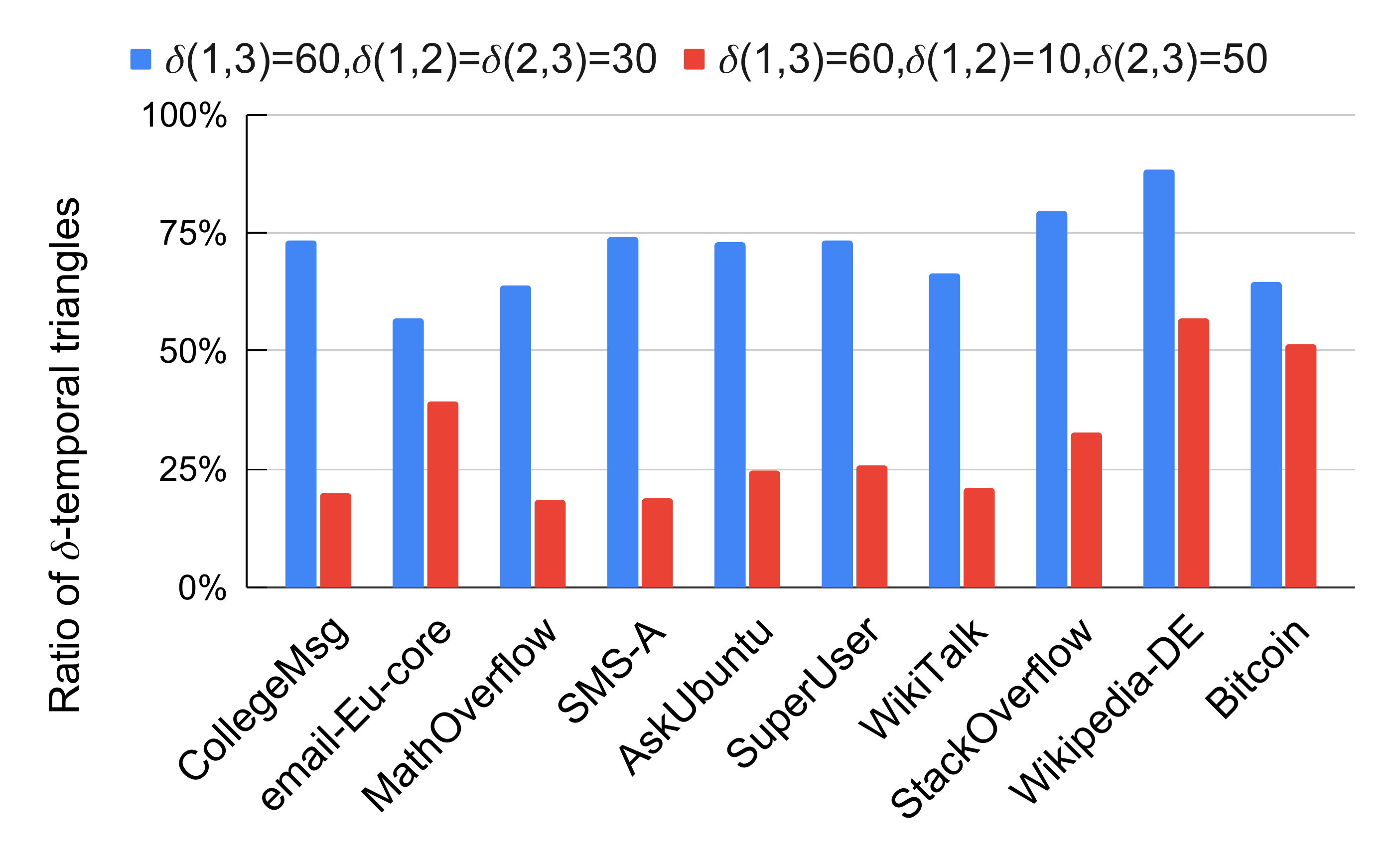}
\caption{Expressivity}
\label{fig:delta_vs_deltas}
\end{subfigure}
\begin{subfigure}[b]{0.33\textwidth}
\centering
\includegraphics[width=\textwidth]{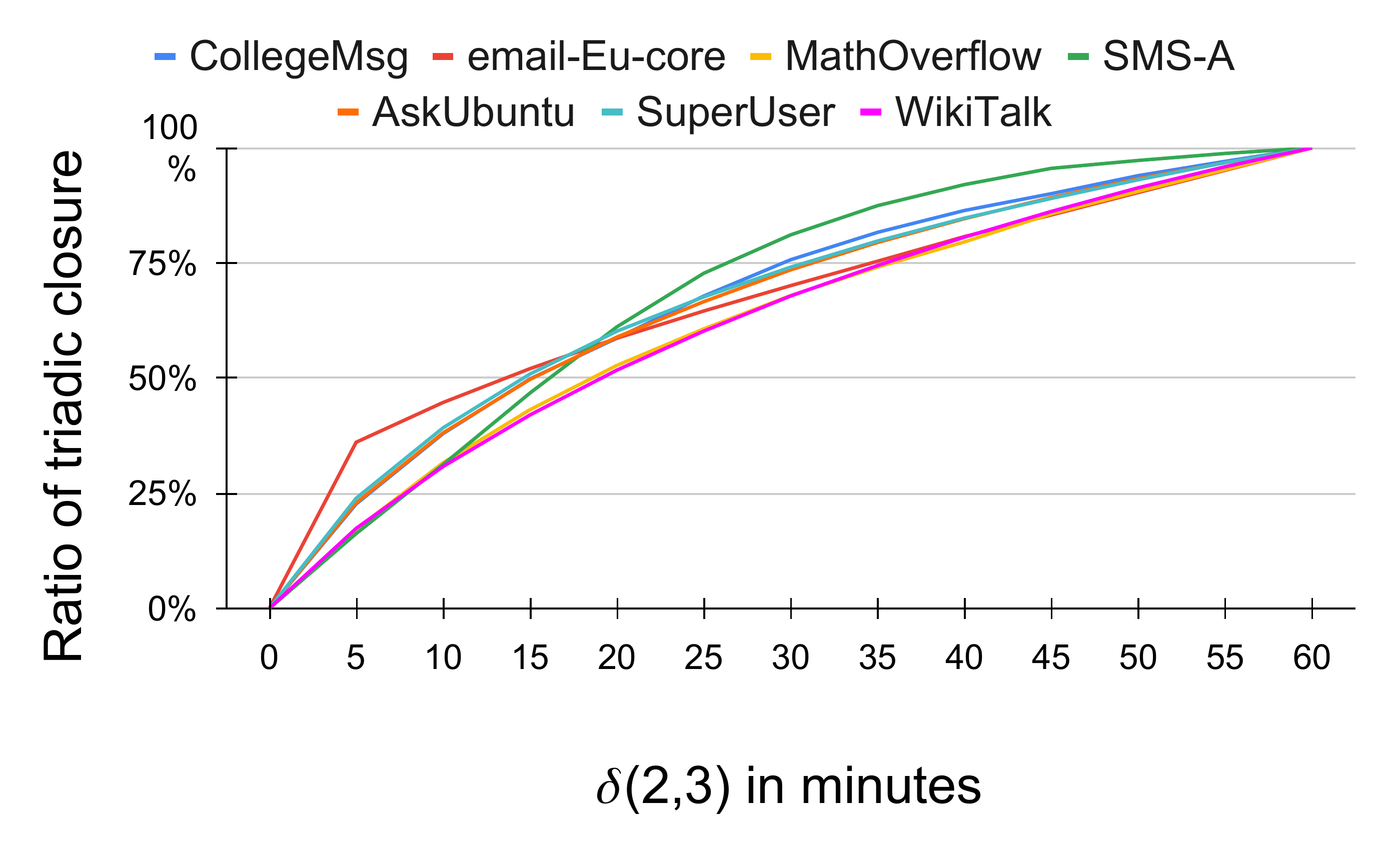}
\caption{Triadic closure over time}
\label{fig:third_edge}
\end{subfigure}
\caption{(a): The Speedup of \DOTTT{} for counting $\threeDelta$-temporal triangles over the \PBL algorithm for counting $\delta_{1,3}$-temporal triangles. (b): We fix $\delta_{1,3}$ to 1 hr. Blue bars show the ratio of (1 hr, 30 mins, 30 mins)-temporal triangles to (1 hr, 1 hr, 1 hr)-temporal triangle. The red bars illustrate the ratio for the case of (1 hr, 10 mins, 50 mins)-temporal triangles and is more restrictive. (c): We fix $\delta_{1,3}=2$ hrs and $\delta_{1,2}=1$ hr. At $t$ we plot the ratio of (2 hrs, 1 hr, t)-temporal triangles to (2 hrs, 1 hr, 1hr)-temporal triangles.}
\label{fig:count_dist_comp}
\end{figure*}

\subsection{Problem Description}
The input is a directed temporal graph $T = (V,E)$. Each edge is a tuple of the form $(u,v,t)$ where $u$ and $v$ are vertices in the temporal graph, and $t$ is a timestamp.
For notational convenience, we assume all timestamps in a temporal network are unique integers.

We introduce our notion of $\threeDelta$-temporal triangles.

\begin{definition} \label{def:temp-tri}
Let $e_1 = (u_1,v_1,t_1), e_2= (u_2,v_2,t_2)$, and $e_3 = (u_3, v_3, t_3)$, be three directed temporal edges where the induced static graph on them is a triangle, and $t_1 < t_2 < t_3$.

$(e_1,e_2,e_3)$ is a \emph{$\threeDelta$-temporal triangle} if $t_2 - t_1 \leq \delta_{1,2}$, $t_3 - t_2 \leq \delta_{2,3}$, and $t_3 - t_1 \leq \delta_{1,3}$.  
\end{definition}

Thus, we specify timestamp differences between \emph{every} pair of edges. When one also considers the direction of edges, there exist eight different types of temporal triangles as shown in~\Fig{temp_tri}. These types are distinguished by temporal ordering of edges and their direction. Thus, for any choice of $\threeDelta$, there are eight different types of temporal triangles (one corresponding to each figure in~\Fig{temp_tri}).

We observe that the notion in \Def{temp-tri} subsumes most existing temporal triangle definitions. Specifically, a $(\delta_{1,3},\delta_{1,3},\delta_{1,3})$-temporal triangle becomes a $\delta_{1,3}$-temporal triangles as defined in PBL~\cite{paranjape2017motifs}.
Temporal triangles with respect to the temporal motif definition by  by Kovanen \etal in~\cite{kovanen2011temporal} consider timestamp differences between consecutive edges in temporal ordering. By our definition, $(2\Delta,\Delta,\Delta)$-temporal triangles capture these types of temporal triangles. Although, the definition in~\cite{kovanen2011temporal} is more restrictive and requires that all edges incident to a node are consecutive events of that node.
Most existing temporal triangle counting literature uses these definitions~\cite{mackey2018chronological,liu2018sampling,sun2019new,wang2020efficient}.

We describe a simple example to see how \Def{temp-tri} offers richer temporal information. Let us measure time in hours, so $(2,1,1)$-temporal triangle is one where the first and second edge (of the triangle) are at most 1 hour apart, and similarly for the second and third edge. Now consider $(1.5, 1, 1)$-temporal triangles. The time gap between the first and second edge (as well as the second and third) is again 1 hour, but the entire triangle must occur within 1.5 hours. There is a significant difference between these cases, but previous definitions of temporal triangles would not distinguish these.

We note that more general temporal motifs, beyond triangles, have been defined. Yet, to the best of our knowledge, most fast algorithms that scale to millions of edges have been designed for triangles. Paranjape~\etal specialized algorithm for  3-edge triangle motifs (temporal triangles) is up to 56x faster than their general motif counting algorithm~\cite{paranjape2017motifs}. Our focus was on scalable algorithms, and hence, on triangle counting. We believe that generalizing \Def{temp-tri} (and our \DOTTT{} algorithm) for general motifs would be compelling future work.




\begin{figure}[th]
\centering
\resizebox{\columnwidth}{!}{
  \begin{tikzpicture}[nd/.style={scale=1,circle,draw,inner sep=2pt},minimum size = 6pt]    
    \matrix[column sep=0.8cm, row sep=0.4cm,ampersand replacement=\&]
    {
    \TriTypeOne \&
    \TriTypeTwo \&
    \TriTypeThree \&
    \TriTypeFour \\
    \TriTypeFive \&
    \TriTypeSix \&
    \TriTypeSeven \&
    \TriTypeEight \\
 };
  \end{tikzpicture}
}
\caption{\label{fig:temp_tri} \small All possible temporal triangle types. The start point of the first edge (in temporal ordering of edges) is shown in red and the end point in green.}
\end{figure}
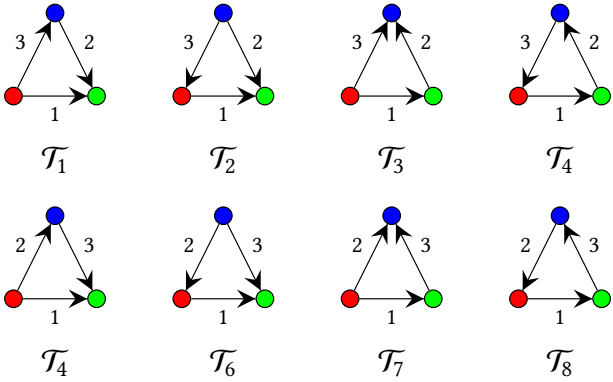

\subsection{Main Contributions} \label{sec:cont}

Our main result is the \emph{Degeneracy Oriented Temporal Triangle Totaler} algorithm, $\DOTTT{}$ that counts $\threeDelta$-temporal triangles as defined in \Def{temp-tri}. The running time is only a logarithmic overhead over static triangle counting. We detail our contributions below.

{\bf Theoretically bridging gap between temporal and static triangle counting:} Our main theorem is the following. 

\begin{theorem}\label{thm:all_runtime}
 Given $\delta_{1,3},\delta_{1,2}$ and $\delta_{2,3}$, the \DOTTT{} algorithm exactly counts each of the eight types of $\threeDelta$-temporal triangles (\Fig{temp_tri}) in a temporal graph in $O(m \degen \log m)$ time. (Here, $m$ is the total number of temporal edges, and $\degen$ is the degeneracy of the underlying static graph.)
\end{theorem}

Observe that, up to a logarithmic factor, our theoretical running time for temporal triangle counting matches the $O(m\degen)$ bound for static triangle counting. As mentioned earlier, the previous best bound was $O(m\numStaticTriangles^{1/2})$. We stress that there is no dependence on the time intervals $\threeDelta$.

The idea of degeneracy orientations is tailored to static graphs, and one of our contributions is to show it can help for temporal triangle counting. A key insight in \DOTTT{} is to process (underlying) static edges in the exact order of the Chiba-Nishizeki algorithm, but carefully consider neighboring edges to capture all temporal triangles. By a non-trivial combinatorial analysis, we can prove that number of times that a temporal edge is processed is upper bounded by $\degen$. We need additional data structure tricks to get the counts efficiently, leading to an extra logarithmic factor.

{\bf Excellent practical behavior of \DOTTT:}
 \DOTTT{} consistently determines temporal triangle counts in less than ten minutes for datasets with tens of millions of edges. We only use a single commodity machine with 64GB memory, without any parallelization. We directly compare \DOTTT{} with the state-of-the-art PBL algorithm. Our algorithm is consistently faster, and as illustrated in~\Fig{speedup} we typically get a factor 1.5 speedup for larger graphs. (We note that \DOTTT{} can count a more general class of temporal triangles.)

We note that for the largest dataset in our experiments, Bitcoin (515.5M edges), \DOTTT{} only uses 64GB memory and runs in less than an hour, while existing methods ran out of memory (details in~\Sec{experimental}).

{\bf Richer triadic information from $\threeDelta$-temporal triangles:}
We demonstrate how $\threeDelta$-temporal triangles can give a richer network
analysis method. Consider~\Fig{delta_vs_deltas}. For a collection of temporal datasets, we generate the counts of (1 hr, 30 min, 30 min)-temporal triangle counts, as well as those for (1 hr, 10 min, 50 min)-temporal triangles. We plot these numbers as a ratio of (1 hr, 1 hr, 1 hr)-temporal triangles. Across the  datasets, the ratios are at most $75\%$. The red bars are typically at most $25\%$, showing the extra power of \Def{temp-tri} in distinguishing temporal triangles.

We note here that for each dataset, \DOTTT{} has the same running time for obtaining the counts for (1 hr, 30 min, 30 min)-temporal triangle and (1 hr, 10 min, 50 min)-temporal triangles, as it has no dependency on the time intervals $\threeDelta$.

An interesting study is presented in~\Fig{third_edge}. The transitivity and clustering coefficients are fundamental quantities of study in network science. In temporal graphs, in addition to these measure, the time it takes for a wedge (2-path) to close could also be of importance. (Zingnani~\etal proposed the triadic closure delay metric that capture the time delay between when a triadic closure is first possible, and when they occur~\cite{zignani2014link}.)
In \Fig{third_edge}, we fix $\delta_{1,3} = $ 2 hrs and $\delta_{1,2}$ = 1 hr. We then vary $\delta_{2,3}$ from zero to 60 minutes, and plot the ratio of $\threeDelta$-temporal triangles to (2hrs, 1hr, 1hr)-temporal triangles. We can see the trends in triadic closure with respect to the time for the third edge. 
We observe that, by and large, half the triangles are formed within 20 minutes of the first two edges appearing. And by 30 minutes, almost $75\%$ of these triangles are formed. These are examples of triadic analyses enabled by \DOTTT.

\subsection{Main challenges} \label{sec:challenge}
In a temporal graph, the number of temporal edges is typically two to three times the number of underlying static edges. Since most triangle counting algorithms are based on some form of wedge enumerations, this leads to a significant increase in the number of edges.
One method used for temporal triangle counting is to simply prune the temporal edges based on the time period~\cite{mackey2018chronological,sun2019tm}. But such algorithms have a dependency on the time period and are inefficient for large 
time periods.

Another significant challenge is the multiplicity of an individual edge can be extremely large. The same edge often occurs many hundreds to thousands of times in a temporal network (in the BitCoin network, there is an edge appearing 447K times~\Tab{runtime}). These edges create significant bottlenecks for enumeration methods. It is not clear how efficient methods on the underlying static graphs (which ignores multiplicities) can help with this problem. Triangle counting often works by finding a wedge (2-path) and checking for the third edge. With multiple temporal edges between the same pair of vertices, this method requires many edge lookups. Paranjape~\etal used a clever idea to process edges on a pair of vertices $O(\numStaticTriangles^{1/2})$ times. The challenge is to bound it by the degeneracy of the graph.

The time constraints expressed by $\threeDelta$-temporal triangles create additional challenges. A clever wedge enumeration exploiting the degeneracy may produce wedges containing the first and second edges of the triangle, the first and third, or the second and third. This makes the lookup (or counting) of possible "matches" for the remaining edge challenging, since it appears we need to look at all multiple edges. On the other hand, if we enumerated wedges that only involved the first and second edge, we cannot benefit for the efficiencies of degeneracy-based methods. Some of these problems can be circumvented for $(\delta,\delta,\delta)$-temporal triangles, but the general case is challenging.

Overall, we can state the main challenge as follows. Fast triangle counting methods (such as degeneracy based methods) necessarily ignore time constraints while generating wedges, making it hard to look for the "closing" edge. On the other hand, a method that exploits the timestamps by (say) pruning cannot get the efficiency gains of degeneracy based methods. One of the insights of \DOTTT{} is a resolution of this tension.



\section{Related Work} \label{sec:related}
There is rich history of work on triangle counting in static graphs. Various algorithm for triangle and motif counting in attributed graphs have also been proposed~\cite{WernickeRasche06,pfeiffer2014attributed,ribeiro2014discovering,mongiovi2018glabtrie,gu2018homogeneous,rossi2019heterogeneous}. Here we only focus on temporal networks and refer the reader to~\cite{al2018triangle} and the tutorial~\cite{seshadhri2019scalable} for a more detailed list of related work.

Graph orientation, in particular degeneracy ordering, is a classic idea in counting triangles and motifs in static graphs, pioneered by Chiba-Nizhizeki~\cite{chiba1985arboricity}. Recently, there has been a number of triangle counting and motif counting algorithms inspired by these techniques~\cite{jha2015path, eden2017approximately, PiSeVi17, jain2017fast, ortmann2017efficient, pashanasangi2020efficiently}. The main benefit of degeneracy ordering is that the out-degree of each vertex becomes small when we orient the static graph based on this ordering.

Kovanen~\etal called two temporal edges $\Delta T$-adjacent if they share a vertex and the difference of their timestamps are at most $\Delta T$~\cite{kovanen2011temporal}. In their definition of temporal motifs, temporal edges must represent consecutive events for a node. Redmond~\etal gave an algorithm for counting $\delta$-temporal motifs but their algorithm does not take  the temporal ordering of edges into account~\cite{redmond2013temporal}, and only counts motifs where incoming edges occur before outgoing edges. Gurukar~\etal present a heuristic for counting temporal motifs~\cite{gurukar2015commit}.

More related to our work, Paranjape-Benson-Leskovec defined the $\delta$-temporal motifs where all edges occur inside a time period $\delta$ and also the temporal ordering of edges are taken into account~\cite{paranjape2017motifs}. They gave a general algorithm for counting $k$-node $\ell$-edge motifs in temporal networks. The main idea behind their algorithm is a moving time window of size $\delta$ over the sequence of all temporal edges for each static motif matching the underlying static motif of the temporal motif of interest. For temporal triangles, their algorithm runs in $O(\tau m)$ time where $\tau$ is the number of triangles in the underlying static graph of the input temporal graph $T$, as it might enumerate temporal edges on static edges with high multiplicity $O(\tau)$ time. They also presented a specialized, more efficient algorithm for counting 3-edge temporal triangles that runs in time $O(\tau^{1/2} m)$. We call their algorithm $\PBL$ and use it as our baseline.

Mackey~\etal presented a backtracking algorithm for counting $\delta$-temporal motifs that maps edges of the motif to the edges of the host graph one by one in temporal (chronological) ordering. For each edge, it only searches through edges that occur in the correct temporal ordering and respect the time gap restriction.~\cite{mackey2018chronological}. Unlike the \PBL algortihm and ours, this algorithm could be inefficient for large values of $\delta$ as its runtime depends on the value of $\delta$.

Liu~\etal~\cite{liu2021temporal} introduced a comparative survey of temporal motif models. Boekhout~\etal gave an algorithm for counting $\delta$-temporal multi-layer temporal motifs~\cite{boekhout2019efficiently}. Li~\etal, developed an algorithm for counting temporal motifs in heterogeneous information networks~\cite{li2018temporal}. Petrovic~\etal gave an algorithm for counting causal paths in time series data on networks~\cite{petrovic2019counting}.

There has also been recent progress on approximating the counts of temporal motifs and triangles~\cite{sun2019new, wang2020efficient}. Particularly , Liu~\etal presented a sampling framework for approximating the counts of $\delta_{1,3}$-temporal motifs~\cite{liu2018sampling}.

\section{Preliminaries}
The input graph is a directed temporal graph that we denote by $T(V,E)$. Let $|V|=n$ and $|E|=m$. Temporal graph $T$ is presented as a collection of $m$ directed temporal edges $e = (u,v,t)$ where $u,v \in V$, and $t$ is the timestamp for edge $e$ where $t\in \RR$. We use $t(e)$ to denote the timestamp of a temporal edge $e$. Note that there could be multiple temporal edges on the same pair of nodes. We assume that all the timestamps in $T$ are unique. This assumption leads to the clean definition of different types of temporal triangles (~\Fig{temp_tri}),  but is not a necessity of our algorithm. To be more specific, our algorithm also works for temporal graph including temporal edges with equal timestamp.

We denote the underlying undirected static graph of $T$ as $G=(V,E_s)$ and put $|E_s| = m_s$. Two vertices in the static graph $G$ are connected if there is at least one temporal edge between them. Formally, $E_s = \{\{u,v\} \mid \exists t: (u,v,t) \in E \vee (v,u,t) \in E\}$. For $v_1,v_2 \in V$, let $\sigma((v_1,v_2))$ denote the temporal multiplicity, that is the number of temporal edges on $\{v_1,v_2\}$ directed from $v_1$ to $v_2$.

As shown in~\Fig{temp_tri}, there are eight different types of temporal triangles. The time restrictions $\delta_{1,3}$, $\delta_{1,2}$, and $\delta_{2,3}$ is not involved in definition of these types and could be applied to each of them. Note that these different types account for all possible ordering of temporal edges in the triangle in addition to their directions.

In a directed graph $G$, we use $N^+(v):\{(v,u) \in E(G)\}$ to denote the out-neighborhood and $N^-(v):\{(u,v) \in E(G)\})$ to denote the in-neighborhood of a vertex $v$. For a vertex $v\in V(G)$, we define out-degree as $d^+(v) = |N^+(v)|$ and in-degree as $d^-(v) = |N^-(v)|$.

Vertex ordering is a central idea in triangle counting and motif analysis in general~\cite{chiba1985arboricity,PiSeVi17,jain2017fast,ortmann2017efficient,turk2019revisiting,pashanasangi2020efficiently,bera2020linear}. Let $G$ be an undirected static simple graph. Given any ordering $\prec$ of $V(G)$,  we can construct a DAG $G_\prec$ by orienting each edge $\{u,v\} \in E(G)$ from $u$ to $v$ iff $u \prec v$.
Next we define \emph{degeneracy} and \emph{degeneracy ordering} formally. 

\begin{definition}
The degeneracy of a graph $G$, denoted by $\degen$, is the smallest integer $k$ such that there exists an ordering $\prec$ of $V(G)$ where $d^+_v(G_\prec) \leq k$ for each $v \in V(G)$.
\end{definition}

\begin{definition}
Degeneracy ordering of an undirected simple static graph $G$ is obtained by removing a vertex with minimum degree repeatedly. The order of the removal of vertices is the degeneracy ordering of $G$. 
\end{definition}

There is an algorithm for finding the degeneracy ordering of a graph $G$ in $O(|E(G)|)$ time~\cite{matula1983smallest}. Let $\prec$ denote the degeneracy ordering.

\section{Main Ideas}
Our algorithm first enumerates static triangles in $G$, the underlying static graph of the input temporal graph $T$. Let $\{u,v,w\}$ be a static triangle. We consider all possible temporal orderings as shown in~\Fig{temp_order}, and all possible orientations as shown in~\Fig{temp_dir}, for a temporal triangle corresponding to $\{u,v,w\}$. A temporal ordering and a temporal orientation together determine the type of the temporal triangle. For example $\tempOrder_1$ and $\dirType_8$ correspond to $\TriType_1$.~\Tab{tri_type_conversion} lists all possible pairs of temporal ordering and orientations and their corresponding type of temporal triangle as a function $\TriTypeConversion$.

We store the input temporal edges of the input temporal graph $T$ in a data structure in the CSR format. Thus, we can assume that we have constant time access to temporal edges on each pair of vertices for each direction in the order of increasing timestamps. Let $\tempOrder$ denote the temporal ordering, and $\dirType$ denote the orientation for which we want to count the temporal triangles. In this section, from here we only consider temporal edges that follow the orientation $\dirType$.

Assume that the timestamps of two of the edges of a temporal triangle corresponding to the static triangle $\{u,v,w\}$ is given. WLOG, assume that these temporal edges correspond to $\{u,v\}$ and $\{u,w\}$. We can use a binary search to find the number of temporal edges on the pair $\{v,w\}$ that are compatible with the two given temporal edges. Note that compatibility of timestamps is determined by the timestamp of edges and the temporal ordering $\tempOrder$. Thus, all we need is to enumerate all possible pairs of temporal edges on $\{u,v\}$ and $\{u,w\}$. This could be an expensive enumeration if both these static edges have high multiplicity of temporal edges.

We show that we can obtain the counts of temporal triangles on $\{u,v,w\}$ without enumeration of all possible pairs of temporal edges on $\{u,v\}$ and $\{u,w\}$. Let $e_1,\ldots,e_{\sigma(\dirType(\{u,v\}))}$ denote the sequence of temporal edges in the order of increasing timestamp on $\{u,v\}$, and $e_1^\prime,\ldots,e_{\sigma(\dirType(\{u,w\}))}^\prime$ denote that of $\{u,w\}$. We enumerate each of these sequences of temporal edges separately, and for each edge we store cumulative counts of compatible temporal edge on $\{v,w\}$. In other words, for each edge $e$ in these two sequence, we store the counts of edges $e_3$ on $\{v,w\}$ that are compatible with $e$ or any other temporal edge in the same sequence with a smaller timestamp.

Then we enumerate the temporal edges on $\{u,v\}$, and for each edge $e_i$ we use binary search to find the sequence of temporal edges $e_j^\prime,\ldots,e_k^\prime$, in increasing order of timestamp, on $\{v,w\}$ that are compatible with $e_i$. We use the cumulative counts of compatible edges on $\{v,w\}$ that we stored for $e_i$, $e_k^\prime$, and $e_j^\prime$ to compute the counts of all temporal triangles on $\{u,v,w\}$ that include $e_i$. 

Although we avoid the enumeration of pairs of temporal edges on $\{u,v\}$ and $\{u,w\}$, our algorithm could still be inefficient. The reason is that static edges $\{u,v\}$ could have high multiplicity of temporal edges and also participate in a large number of static triangles. Here is where we use the power of vertex ordering and graph orientation techniques.

\DOTTT{} enumerates static triangles in $G_\prec$, and when processing a static triangle $\{u,v,w\}$ where $u$ comes first in the degeneracy ordering $\prec$ of $G$, it only enumerates temporal edges on $\{u,v\}$ and $\{u,w\}$. Thus, each temporal edge on a pair $\{x,y\}$ where $x \prec y$, is processed only for static triangles where the third vertex is in the out-neighborhood of $x$. But we know that the out-degree of each vertex is bounded by $\degen$ in $G_\prec$. Therefore, each such temporal edge on $\{x,y\}$ is processed $O(\degen)$ times.

\begin{figure}[t]
\centering
\resizebox{0.9\columnwidth}{!}{
  \begin{tikzpicture}[nd/.style={scale=1,circle,draw,inner sep=2pt},minimum size = 6pt]    
    \matrix[column sep=1.0cm, row sep=0.4cm,ampersand replacement=\&]
    {
    \OrderOne \&
    \OrderTwo \&
    \OrderThree \\
    \OrderFour \&
    \OrderFive \&
    \OrderSix \\
 };
  \end{tikzpicture}
}
\caption{\label{fig:temp_order} All possible ordering of temporal edges of a temporal triangle corresponding to a static triangle $\{u,v,w\}$.}
\end{figure}
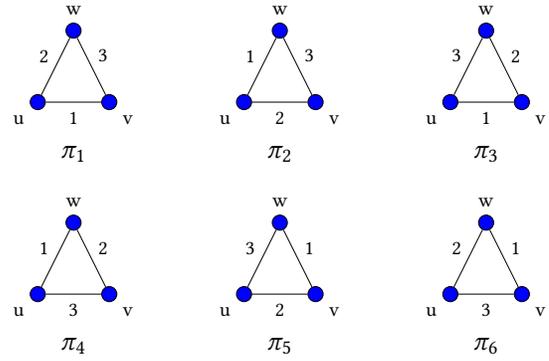

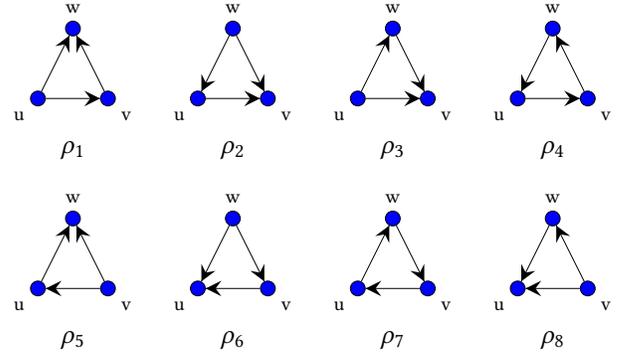
\begin{figure}[t]
\centering
\resizebox{\columnwidth}{!}{
  \begin{tikzpicture}[nd/.style={scale=1,circle,draw,inner sep=2pt},minimum size = 6pt]    
    \matrix[column sep=0.4cm, row sep=0.4cm,ampersand replacement=\&]
    {
    \DirOne \&
    \DirTwo \&
    \DirThree \&
    \DirFour \\
    \DirFive \&
    \DirSix \&
    \DirSeven \&
    \DirEight \\
 };
  \end{tikzpicture}
}
\caption{\label{fig:temp_dir} All possible orientations of temporal edges of a temporal triangle corresponding to a static triangle $\{u,v,w\}$.}
\end{figure}

\begin{table}[t]
\caption{Conversion from temporal ordering and orientation to temporal triangle type.} 
\begin{center}
\begin{tabular}{c|cccccccc}
$\TriTypeConversion$ & $\dirType_1$ & $\dirType_2$ & $\dirType_3$ & $\dirType_4$ & $\dirType_5$ & $\dirType_6$ & $\dirType_7$ & $\dirType_8$\\\hline  
$\tempOrder_1$ & $\TriType_7$ & $\TriType_6$ & $\TriType_5$ & $\TriType_8$ & $\TriType_3$ & $\TriType_2$ & $\TriType_4$ & $\TriType_1$ \\
$\tempOrder_2$ & $\TriType_5$ & $\TriType_3$ & $\TriType_7$ & $\TriType_4$ & $\TriType_6$ & $\TriType_1$ & $\TriType_8$ & $\TriType_2$ \\
$\tempOrder_3$ & $\TriType_3$ & $\TriType_2$ & $\TriType_1$ & $\TriType_4$ & $\TriType_7$ & $\TriType_6$ & $\TriType_8$ & $\TriType_5$ \\
$\tempOrder_4$ & $\TriType_1$ & $\TriType_7$ & $\TriType_3$ & $\TriType_8$ & $\TriType_2$ & $\TriType_5$ & $\TriType_4$ & $\TriType_6$ \\
$\tempOrder_5$ & $\TriType_6$ & $\TriType_1$ & $\TriType_2$ & $\TriType_8$ & $\TriType_5$ & $\TriType_3$ & $\TriType_4$ & $\TriType_7$ \\
$\tempOrder_6$ & $\TriType_2$ & $\TriType_5$ & $\TriType_6$ & $\TriType_4$ & $\TriType_1$ & $\TriType_7$ & $\TriType_8$ & $\TriType_3$ \\
\end{tabular}
\end{center}
\label{tab:tri_type_conversion}
\end{table}

\section{Our Main Algorithm}\label{sec:mainalgorithm}
In this section we describe our algorithm for getting $\threeDelta$-temporal triangles counts. Let $T=(V,E)$ be the input directed temporal graph given as a list of temporal edges sorted by timestamps. Although not necessary for our algorithm, assuming that edges are given in increasing order of timestamp is common in temporal networks as the edges are recorded in their order of occurrence~\cite{paranjape2017motifs}.

We first extract the static graph $G(V,E_s)$ from $T$. Then, we obtain the degeneracy ordering of $G$, denoted by $\prec$ using the algorithm by Matula and Beck~\cite{matula1983smallest}, and orient the edges of $G$ with respect to $\prec$ to get the DAG $G_\prec$.
We start by enumerating static triangles in $G_\prec$. This can be done in $O(m_s\degen)$ where $\degen$ is the degeneracy of  $G$~\cite{chiba1985arboricity,PiSeVi17}.

Note that all triangles in $G_\prec$ are acyclic as $G_\prec$ is a DAG, so each triangle in $G$ correspond to an acyclic triangle in $G_\prec$. In order to enumerate all triangles in $G$, we enumerate all directed edges in $G_\prec$, and for each directed edge $(u,v)$ we enumerate $N^+(u)$. For each vertex $w \in N^+(u)$, we check whether $\{u,v,w\}$ is a triangle by checking the existence of an edge between $v$ and $w$.
 
We call vertex $u$ in a static triangle $\{u,v,w\}$ the \emph{source vertex} if $u \prec v$ and $u \prec w$. Let $\{u,v,w\}$ be the triangle being processed while enumerating triangles in $G_\prec$. WLOG, assume $u$ is the source vertex in $\{u,v,w\}$. Thus, the number of times we visit  $\{u,v\}$ or $\{u,w\}$ in a static triangle are limited by $d^+_{G_\prec}(u)$ that is bounded by $\degen$. But the number of times we visit the static edge $\{v,w\}$ is not bounded by $\degen$, so we want to avoid enumerating temporal edges on $\{v,w\}$. Next, we show how to count the number of $\threeDelta$-temporal triangles corresponding to the static triangle $\{u,v,w\}$.
 
We define the temporal ordering of a temporal triangle corresponding to a static triangle $\{u,v,w\}$ as a mapping $\tempOrder: \{1,2,3\} \rightarrow \{\{u,v\},\{u,w\},\{v,w\}\}$. There are six different possible temporal orderings as shown in~\Fig{temp_order}.

We define the orientation of a temporal triangle corresponding to a static triangle $\{u,v,w\}$ as a mapping $\dirType$ from each pair of vertices of $\{u,v,w\}$ to one of the two possible ordered pairs of the same pair of vertices. For example, $\dirType_1(\{u,v\})=(u,v)$ for $\dirType_1$ in~\Fig{temp_dir}. The orientation of a temporal triangle simply determines the direction of its temporal edges. Each such temporal edge can take two possible directions, so there are eight types of orientation such a temporal triangle can take as shown in~\Fig{temp_dir}. Note that orientation of a temporal triangle is independent of its temporal ordering.

It is easy to see that the temporal ordering and orientation determine the type of the temporal triangle. But different combinations of temporal orderings and orientation could result in the same type. The temporal triangle type for all possible pairs of temporal ordering and orientation are shown in~\Tab{tri_type_conversion}.

For a temporal ordering $\tempOrder$ and for $i \in \{1,2,3\}$, we use $S_i(\tempOrder, \dirType)$ to denote the sequence of temporal edges between the pair of vertices $\tempOrder(i)$ that have the direction $\dirType(\tempOrder(i))$, in sorted order of timestamp. When $\tempOrder$ and $\dirType$ are clear from the context, we use $S_i$ instead of $S_i(\tempOrder,\dirType)$. We assume that we have access to $S_1$, $S_2$, and $S_3$ in constant time. Let $\sigma_i$ denote the length of $S_i$. We use $S_i[\ell]$ to denote the $\ell$-th edge in the sequence $S_i$, and $S_i[\ell:\ell^\prime]$ to denote the consecutive subsequence of $S_i$ ranging from $S_i[\ell]$ to $S_i[\ell^\prime]$.

For a sequence $S$ of temporal edges in increasing order of timestamp and timestamps $t$ and $t^\prime$ where $t \leq t^\prime$, let $\edgeCount([t,t^\prime], S)$ denote the number of edges in $S$ with a timestamp in the time window $[t,t^\prime]$. For given $\delta_{1,3},\delta_{1,2}$, and $\delta_{2,3}$, let $\consecutivetriangleCount(\{u,v,w\}, \tempOrder, \dirType)$ denote the number of $\threeDelta$-temporal triangles corresponding to the static triangle $\{u,v,w\}$, temporal ordering $\tempOrder$, and orientation $\dirType$.

\begin{lemma} \label{lem:delta_temporal_count}
For a static triangle $\{u,v,w\}$, a temporal ordering $\tempOrder$, and an orientation $\dirType$,
\begin{align*}
    & \consecutivetriangleCount(\{u,v,w\}, \tempOrder, \dirType) = \sum_{e_2 \in S_2} \sum\limits_{\substack{e_1 \in S_1 \\ t(e_1) \in [t(e_2)-\delta_{1,2},t(e_2)]}} & \\ &\edgeCount([t(e_2),\min(t(e_2)+\delta_{2,3},t(e_1)+\delta_{1,3})],S_3)&
\end{align*}
\end{lemma}
 
\begin{proof}
If temporal edge $e_1 \in S_1$ is in a $\threeDelta$-temporal triangle with edge $e_2 \in S_2$, then $t(e_1) \in [t(e_2)-\delta_{1,2},t(e_2)]$. Fix a pair of temporal edges $(e_1,e_2)$ in $S_1 \times S_2 = \{(e_1,e_2) \mid e_1 \in S_1 \wedge e_2 \in S_2\}$ where $t(e_2) \in [t(e_1), t(e_1) + \delta_{1,2}]$. A temporal edge $e_3 \in S_3$ composes a $\threeDelta$-temporal triangle with $e_1$ and $e_2$ iff $t(e_3) \in [t(e_2), \min(t(e_2) + \delta_{2,3}, t(e_1) + \delta_{1,3})]$.
\end{proof}

For a triangle $\{u,v,w\}$ where $u$ is the source vertex, we divide all six possible temporal orderings into three categories based on the place of $\{v,w\}$ in them. Recall that we want to avoid enumerating temporal edges on $\{v,w\}$. In $\tempOrder_1$ and $\tempOrder_2$, $\{v,w\}$ is assigned to the third place. $\{v,w\}$ is assigned to the second place in $\tempOrder_3$ and $\tempOrder_4$, and finally to the first place in temporal ordering $\tempOrder_5$ and $\tempOrder_6$.

{\bf Temporal orderings $\tempOrder_1$ and $\tempOrder_2$:}
Using~\Lem{delta_temporal_count}, one can compute $\consecutivetriangleCount(\{u,v,w\}, \tempOrder, \dirType)$ by enumerating pairs of temporal edges in $S_1 \times S_2 = \{(e_1,e_2) \mid e_1 \in S_1 \wedge e_2 \in S_2\}$. For each pair we compute $\edgeCount([t(e_2), \min(t(e_2) + \delta_{2,3}, t(e_1) + \delta_{1,3}),S_3)$ using binary search. To get the final counts we sum $\edgeCount([t(e_2), \min(t(e_2) + \delta_{2,3}, t(e_1) + \delta_{1,3}),S_3)$ over all pairs $(e_1,e_2) \in S_1 \times S_2$. But enumerating $S_1 \times S_2$ could be expensive and this process overall runs in time $O(\sigma_1\sigma_2 \log(\sigma_3))$. Next, we show that we can compute the same count by enumerating edges in $S_1$ and $S_2$ separately and storing cumulative counts of compatible edges on $S_3$ for each edge.
 
For $i,j \in \{1,2,3\}$ where $i \neq j$, and  $\ell, \ell^\prime \in \{1,\ldots,\sigma_i\}$ where $\ell \leq \ell^\prime$ we use $\cumulativeEdgeCount_{+\delta_{1,3}}(S_i[\ell:\ell^\prime],S_j)$ to denote the cumulative count of edges in $S_j$ with a timestamp in $[t(e),t(e)+\delta_{1,3}]$ for edges $e$ in the sequence $S_i[\ell:\ell^\prime]$. Formally
\begin{align*}
\cumulativeEdgeCount_{+\delta_{1,3}}(S_i[\ell:\ell^\prime],S_j) = \sum\limits_{\ell \leq r \leq \ell^\prime} \edgeCount([t(S_i[r]),t(S_i[r])+\delta_{1,3}], S_j).
\end{align*}

Cumulative counts $\cumulativeEdgeCount_{\infty}$, $\cumulativeEdgeCount_{-\delta_{1,3}}$, and $\cumulativeEdgeCount_{-\infty}$, are defined the same way with time intervals $[t(e),\infty)$, $[t(e)-\delta_{1,3},t(e)]$, and $(-\infty,t(e)]$, respectively. $\cumulativeEdgeCount_{-\delta_{1,2}}$, $\cumulativeEdgeCount_{+\delta_{1,2}}$, $\cumulativeEdgeCount_{-\delta_{2,3}}$, and $\cumulativeEdgeCount_{+\delta_{2,3}}$ are defined similarly.
Note that we can compute $\cumulativeEdgeCount(S_i[1:\ell],S_j)$ for each $\ell \in \{1,\ldots,\sigma_i\}$ with one pass over $S_i$, and once we have these counts, we can get the cumulative counts $\cumulativeEdgeCount(S_i[\ell^\prime:\ell^{\prime\prime}],S_j)$,  for each consecutive subsequence $S_i[\ell^\prime:\ell^{\prime\prime}]$ of $S_i$ as follows.
\begin{align*}
\cumulativeEdgeCount_{+\delta_{1,3}}(S_i[\ell^\prime:\ell^{\prime\prime}],S_j) & = \cumulativeEdgeCount_{+\delta_{1,3}}(S_i[1:\ell^{\prime\prime}],S_j) \\ & - \cumulativeEdgeCount_{+\delta_{1,3}}(S_i[1:\ell^\prime-1],S_j).
\end{align*}
where $\cumulativeEdgeCount_{+\delta_{1,3}}(S_i[1:0],S_j) = 0$.

We first enumerate edges in $S_1$.
For each edge $e_1 \in S_1$ we compute $\cumulativeEdgeCount_{+\delta_{1,3}}(S_1[1:\ell],S_3)$ and $\cumulativeEdgeCount_{\infty}(S_1[1:\ell],S_3)$ for each $\ell \in \{1,\ldots,\sigma_1\}$ and store them for $e_1$.
Next, we enumerate edges in $S_2$ and compute $\cumulativeEdgeCount_{\infty}(S_2[1:\ell],S_3)$ for each $\ell \in \{1,\ldots,\sigma_2\}$.

Fix an edge $e_2 \in S_2$. Let $\ell_{\first}$ and $\ell_{\last}$ be the indices of the first and last edges in $S_1$ with a timestamp in $[t(e_2)-\delta_{1,2},t(e_2)]$. Also let $\ell_{\delta_{2,3}}$ be the index of the last edge in $S_1$ with a timestamp at most $t(e_2)-\delta_{1,3}+\delta_{2,3}$. We can find $\ell_{\first},\ell_{\delta_{2,3}}$, and $\ell_{\last}$ using a binary search on $S_1$. Note that $\delta_{1,3} \leq \delta_{1,2} + \delta_{2,3}$, thus $\ell_{\first} \leq \ell_{\delta_{2,3}} \leq \ell_{\last}$.

First consider the temporal edges $S_1[i]$ where $\ell_{\first} \leq i \leq \ell_{\delta_{2,3}}$. For any such edge $t(S_1[i])+\delta_{1,3} \leq t(e_2)+\delta_{2,3}$, so the timestamp of compatible edges in $S_3$ lie in the interval $[t(e_2),t(S_1[i])+\delta_{1,3}]$. Having stored the cumulative counts described above, we can compute the number of pairs of temporal edges $(e_1,e_3) \in S_1[\ell_{\first}:\ell_{\delta_{2,3}}] \times S_3$ that compose a $\threeDelta$-temporal triangle on $\{u,v,w\}$ with $e_2$, complying with $\tempOrder_1$ and $\dirType$, as follows.
\begin{align*}
    & \sum\limits_{\ell \leq i \leq \ell_{\delta_{2,3}}} \edgeCount([t(e_2),t(S_1(i))+\delta_{1,3}], S_3) = & \\
    & \cumulativeEdgeCount_{+\delta_{1,3}}(S_1[\ell_{\first}:\ell_{\delta_{2,3}}], S_3) - \cumulativeEdgeCount_{\infty}(S_1[\ell_{\first}:\ell_{\delta_{2,3}}],S_3) & \\
     + & (\ell_{\delta_{2,3}}-\ell_{\first}+1) \cdot \edgeCount([t(e_2),\infty),S_3)  & 
\end{align*}

Now, we count the number of temporal edges in $S_3$ that compose a triangle with $e_2$ and $S_1[i]$, where $\ell_{\delta_{2,3}} < i \leq \ell_{\last}$. For a temporal edge $S_1[i]$ where $\ell_{\delta_{2,3}} < i \leq \ell_{\last}$, we have $t(S_1[i]) + \delta_{1,3} > t(e_2)+\delta_{2,3}$. Thus, there are $\edgeCount([t(e_2), t(e_2)+\delta_{2,3}], S_3)$ edges on $S_3$ that compose a triangle with $e_2$ and $S_1[i]$. So the final count of pairs $(e_1,e_3) \in S_1 \times S_3$ that are in a temporal triangle with $e_2$ corresponding to the static triangle $\{u,v,w\}$ can be computed as follows. 
\begin{align*}
& \sum\limits_{\substack{e_1 \in S_1, t(e_1) \in \\ [t(e_2)-\delta_{1,2},t(e_2)]}} \edgeCount([t(e_2),\min(t(e_2)+\delta_{2,3},t(e_1)+\delta_{1,3})],S_3) & \\
    = & \sum\limits_{\ell_{\first} \leq i \leq \ell_{\delta_{2,3}}} \edgeCount([t(e_2),t(S_1(i))+\delta_{1,3}],S_3) &  \\
    + & (\ell_{\last} - \ell_{\delta_{2,3}}) \cdot \edgeCount([t(e_2),t(e_2)+\delta_{2,3}],S_3) &  \\
\end{align*}

By~\Lem{delta_temporal_count}, to get $\consecutivetriangleCount(\{u,v,w\}, \tempOrder, \dirType)$, we only need to sum these counts over edges in $S_2$. Let $\langle u,v,w \rangle$ denote a static triangle where $u \prec v \prec w$.~\Alg{tempOrder_1_count} formalizes the procedure described above for computing $\consecutivetriangleCount(\langle u,v,w \rangle, \tempOrder, \dirType)$ where $\tempOrder$ is either $\tempOrder_1$ or $\tempOrder_2$.

\begin{algorithm}[th]
\setstretch{1.2}
\caption{Counting $\threeDelta$-temporal triangles corresponding to a static triangle and temporal orientation $\tempOrder_1$ or $\tempOrder_2$}\label{alg:tempOrder_1_count}
\begin{algorithmic}[1]
\Procedure{TTC-vw3}{$\delta_{1,3}$, $\delta_{1,2}$, $\delta_{2,3}$,$\langle u,v,w \rangle$, $\tempOrder$, $\dirType$}
\LeftComment $\tempOrder(\{v,w\})=3$
\State Enumerate $S_1$ and compute $\cumulativeEdgeCount_{+\delta_{1,3}}$ and $\cumulativeEdgeCount_{\infty}$ on $S_3$.
\State count = 0
\For{$i= 1,\ldots,\sigma_2$}
\State Let $\ell_{\first} = \textsc{lowerBound}(t(S_2[i])-\delta_{1,2}, S_1)$
\State Let $\ell_{\delta_{2,3}} = \textsc{upperBound}(t(S_2[i])-\delta_{1,3}+\delta_{2,3}, S_1)$
\State Let $\ell_{\last}= \textsc{upperBound}(t(S_2[i]), S_1)$
\LeftComment Edges in $S_1[\ell_{\first}:\ell_{\delta_{2,3}}]$
\State count $+= \cumulativeEdgeCount_{+\delta_{1,3}}(S_1[\ell_{\first}:\ell_{\delta_{2,3}}],S_3)$
\State count $-= \cumulativeEdgeCount_{\infty}(S_1[\ell_{\first}:\ell_{\delta_{2,3}}],S_3)$
\State count $+= (\ell_{\delta_{2,3}}-\ell_{\first}+1) \cdot \edgeCount([t(S_2[i]),\infty), S_3)$
\LeftComment Edges in $S_1[\ell_{\delta_{2,3}}+1:\ell_{\last}]$
\State count $+= (\ell_{\last} - \ell_{\delta_{2,3}}) \cdot \edgeCount([t(S_2[i]), t(S_2[i])+\delta_{2,3}], S_3)$
\EndFor
\State \Return count
\EndProcedure
\end{algorithmic}
\end{algorithm}

The algorithms for the category of $\tempOrder_3$ and $\tempOrder_4$ and the category of $\tempOrder_5$ and $\tempOrder_6$ are similar to~\Alg{tempOrder_1_count}, so we defer these to~\Apx{missing_algo}. Here, we suffice to say that similar to~\Alg{tempOrder_1_count}, in algorithms for the remaining temporal orderings, for each static triangle $\{u,v,w\}$ where $u$ is the source vertex, we only enumerate temporal edges on $\{u,v\}$ and $\{u,w\}$. And for each temporal edge we perform a constant number of binary searches on temporal edges on the other two static edges of the static triangle $\{u,v,w\}$.

\subsection{Getting the Counts for All Temporal Triangle Types}

Now that we can count $\threeDelta$-temporal triangles for each combination of temporal ordering and orientation, it only remains to get the counts for each temporal triangle type (~\Fig{temp_tri}). Let $\TriTypeConversion(\tempOrder,\dirType)$ denote the triangle type for $\tempOrder$ and $\dirType$.~\Alg{count_all_delta} gets the counts for all eight types. Now, we can finally prove~\Thm{all_runtime}.

\begin{algorithm}[th]
\caption{Counting $\threeDelta$-temporal triangles for each temporal triangle types}\label{alg:count_all_delta}
\begin{algorithmic}[1]
\Procedure{Count-temporal-triangles}{$T$, $\delta_{1,3}$, $\delta_{1,2}$, $\delta_{2,3}$}
\State Extract the static graph $G$ of $T$.
\State Find the degeneracy ordering $\prec$ of $G$.
\State Derive $G_\prec$ by orienting $G$ with respect to $\prec$.
\State Initialize Counts to 0 for $\TriType_1,\ldots,\TriType_8$.
\ForAll{Static triangles $\{u,v,w\}$}
\Statex\Comment{WLOG let $u\prec v \prec w$}
\ForAll{Temporal ordering $\tempOrder$ and orientation $\dirType$}
\State Counts($\TriTypeConversion(\tempOrder,\dirType)$) += $\consecutivetriangleCount(\langle u,v,w \rangle, \tempOrder, \dirType)$
\Comment{~\Tab{tri_type_conversion}}
\Statex\Comment{Using~\Alg{tempOrder_1_count},~\Alg{tempOrder_3_count}, and~\Alg{tempOrder_5_count}}
\EndFor
\EndFor
\EndProcedure
\end{algorithmic}
\end{algorithm}




 
\begin{proof}[Proof of Theorem~\ref{thm:all_runtime}]

Extracting the static graph $G$ from $T$ can be done in $O(m)$ time. We simply enumerate all temporal edges of $T$ and for each temporal edge $e=(v_1,v_2,t)$, we add an static edge between $\{v_1,v_2\}$ in $G$ if they are not connected already. The degeneracy ordering of $G$ could be obtained in $O(m_s)$ time~\cite{matula1983smallest}, and $G_\prec$ could also be derived in time $O(m_s)$. For enumerating static triangles we first enumerate each edge in $G_\prec$. For each edge $(u,v)$, we enumerate $N^+(u)$ which takes $O(\degen)$ as $d^+_{G_\prec}(u) \leq \degen$. We can lookup if there is an edge between $v$ and $w$ in constant time. Thus, enumerating triangles take $O(m_s \cdot \degen)$ time overall.
 
Note that for each static triangle we only enumerate temporal edges on static edges incident to the source vertex. So for the static triangle $\langle u,v,w \rangle$, we only enumerate temporal edges on the pairs $\{u,v\}$ and $\{u,w\}$. While processing a temporal edge during enumeration of temporal edges on $\{u,v\}$ or $\{u,w\}$, we either perform a constant time operation, or spend $O(\log(\sigma_{\max}))$ time for a constant number of binary searches over the temporal edges of the other two static edges in the static triangle $\langle u,v,w \rangle$. Thus,
\begin{align*}
    T(\mcA) = O(m_s \cdot \degen + \sum\limits_{\langle u,v, w \rangle} (\sigma(u,v) + \sigma(u,w)) \log(\sigma_{\max}))
\end{align*}
where $\mcA$ denotes~\Alg{count_all_delta}, and $T(\mcA)$ denotes the worst case time complexity of $\mcA$.

For each vertex $u \in V$, $d^+_{G_\prec}(u) \leq \degen$, so each edge $(u,v)$ in $G_\prec$, is a part of at most $\degen$ static triangles where $u$ is the source vertex. Therefore, the temporal edges on each edge $\{u,v\}$ in $G$ are enumerated at most $O(\degen)$ times. Thus, 
\begin{align*}
    T(\mcA) = O(m_s \cdot \degen + \sum\limits_{\{u,v\} \in E_s} (\sigma(u,v) + \sigma(v,u)) \cdot  \degen \log(\sigma_{\max})).
\end{align*}
Hence,
\begin{align*}
    T(\mcA) = O(m \degen \log(\sigma_{\max})).
\end{align*}
\end{proof}

\section{Experimental Evaluations}\label{sec:experimental}
\begin{table*}[!ht]
\scriptsize
\caption{Descriptions of the datasets and runtime of \DOTTT{} and \PBL.} 
\begin{center}
\begin{tabular}{l|rrrrrrr|rr}
\hline
\textbf{dataset} & $\#$\textbf{vertices} & $\#$\textbf{edges} & $\#$\textbf{static edges} & $\#$\textbf{static triangles} & \textbf{degeneracy} & \textbf{max multiplicity} & \textbf{time span (years)} & \textbf{\DOTTT{} runtime} & \textbf{\PBL runtime}\\\hline
CollegeMsg & 1.9K &	59.8K & 13.8K & 14.3K & 20 & 98 & 0.51 & 0.09 & 0.07\\
email-Eu-core & 986 & 332K & 16.1K & 105K & 34 & 2.8K & 2.2 & 2.31 & 3.37\\
MathOverflow & 24.7K & 390K & 188K & 1.4M & 78 & 225 & 6.46 & 3.17 & 3.6\\
SMS-A & 44.1K & 545K & 52.22K & 10K & 9 & 5.3K & 0.92 & 0.45 & 0.81\\
AskUbuntu & 157K & 727K & 456K & 680K & 48 & 154 & 7.09 & 2.23 & 5.08\\
SuperUser & 192K & 1.11M & 715K & 1.54M & 61 & 78 & 7.59 & 4.41 & 8.84\\
WikiTalk & 1.09M & 6.11M & 2.79M & 8.12M & 124 & 1.1K & 6.21 & 34 & 56\\
StackOverflow & 2.58M & 47.9M & 28.18M & 114.2M & 198 & 549 & 7.60 & 347 & 678\\
Wikipedia-DE & 2.17M & 86.21M & 39.71M & 169.9M & 265 & 347 & 10.18 & 576 & 987\\
Bitcoin & 59.61M & 515.5M & 366.4M & 706.2M & 604 & 447K & 5.98 & 2923 & 4374\\
\hline
\end{tabular}
\end{center}
\label{tab:runtime}
\end{table*}

\begin{figure*}[!ht]
\centering
\begin{subfigure}[b]{0.33\textwidth}
\centering
\includegraphics[width=\textwidth]{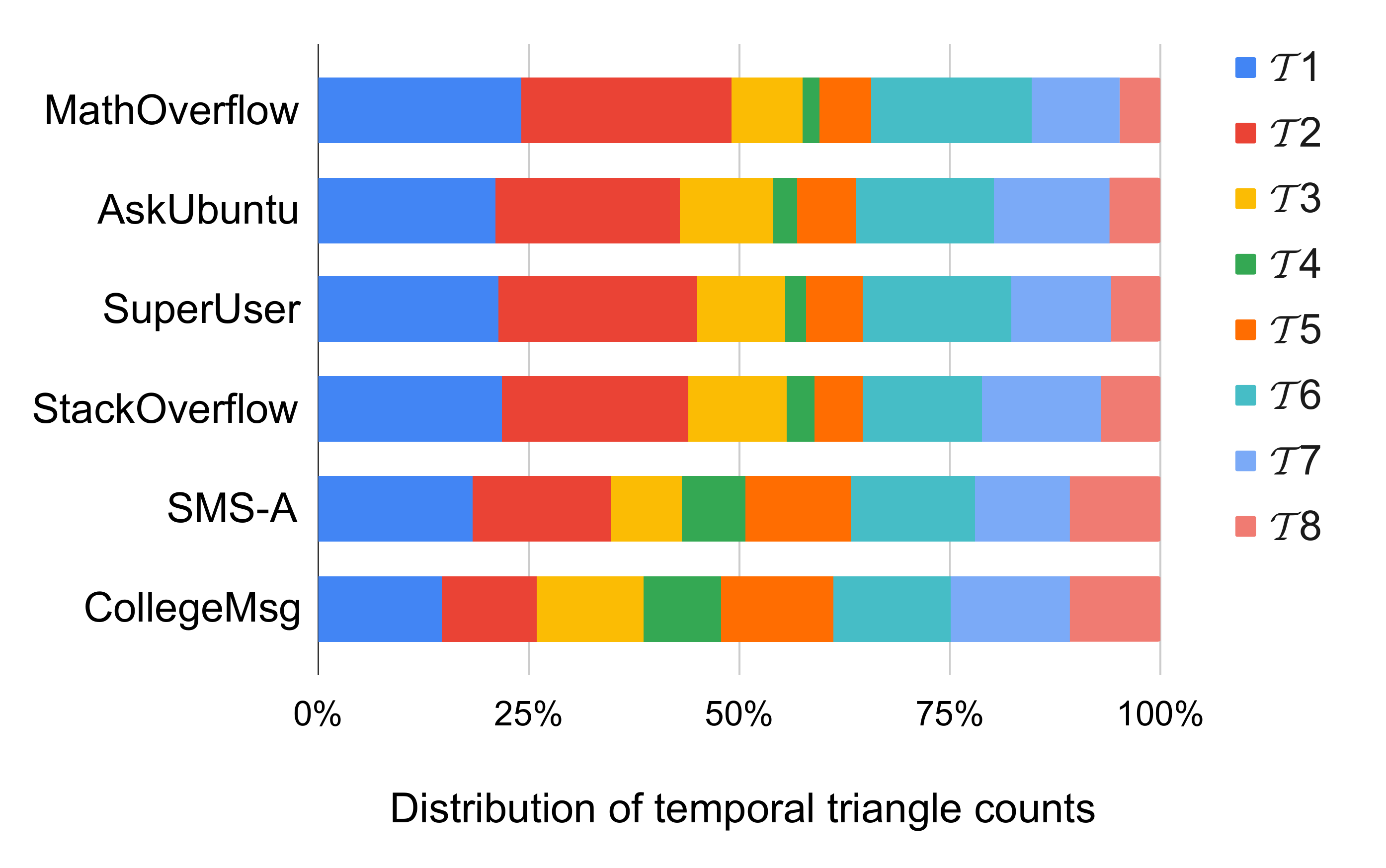}
\caption{Temporal triangle count distribution}
\label{fig:count_dist}
\end{subfigure}
\begin{subfigure}[b]{0.33\textwidth}
\centering
\includegraphics[width=\textwidth]{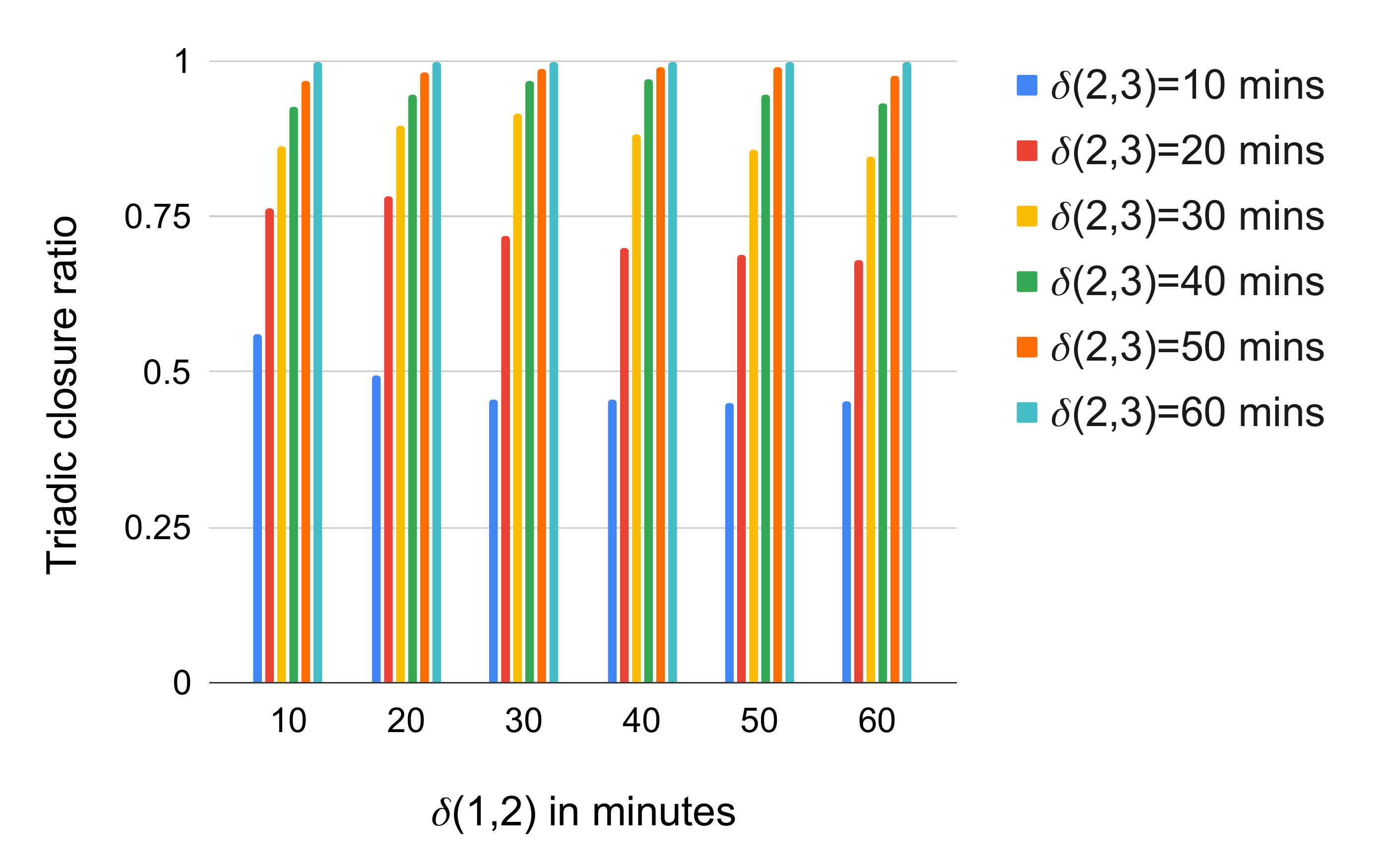}
\caption{Effect of $\delta_{1,2}$ and $\delta_{2,3}$ on triadic closure}
\label{fig:closure_delta1}
\end{subfigure}
\begin{subfigure}[b]{0.33\textwidth}
\centering
\includegraphics[width=\textwidth]{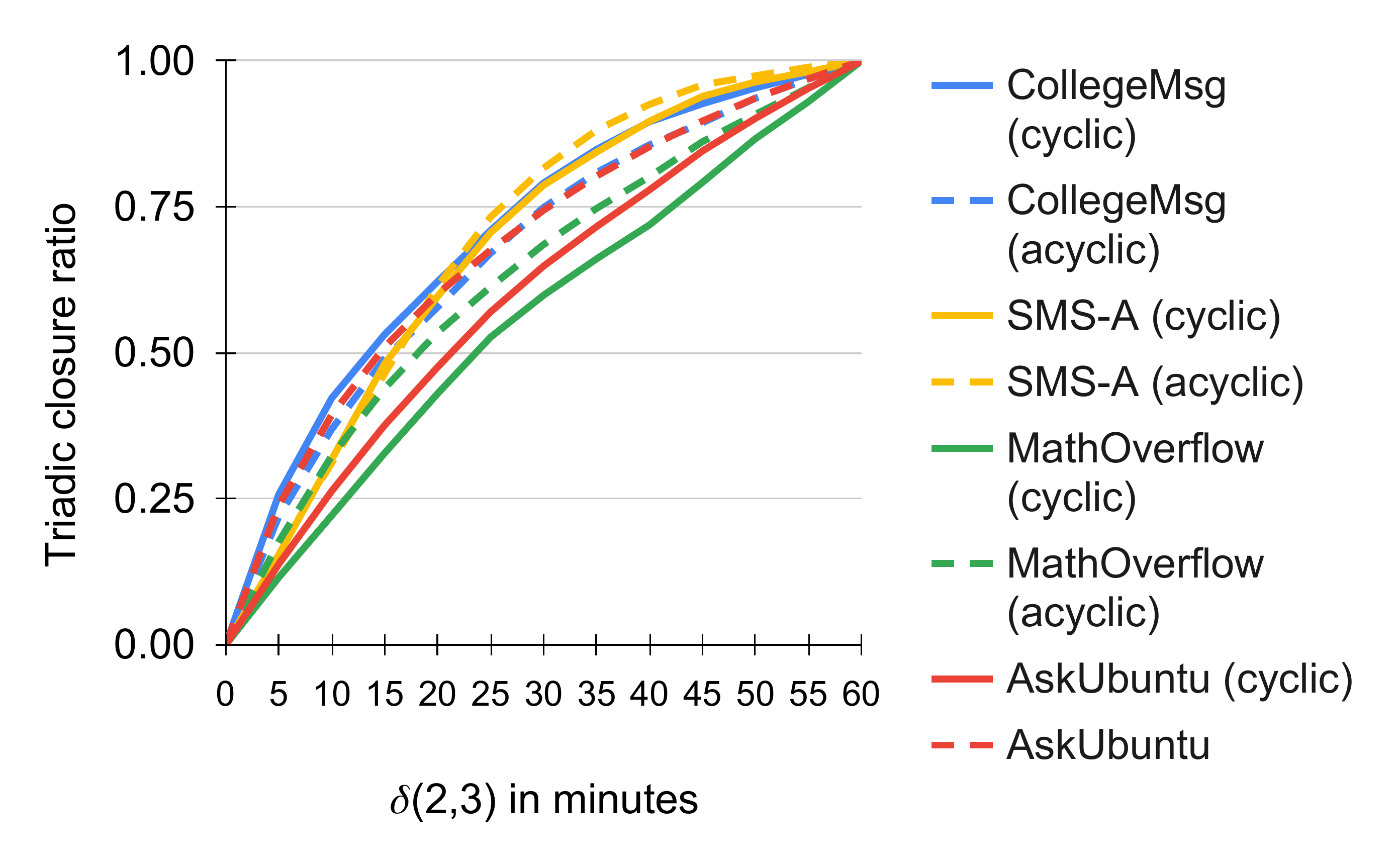}
\caption{Triadic closure in cyclic and acyclic cases}
\label{fig:closure_cyclic}
\end{subfigure}
\caption{(a):The distribution of (1 hr, 1 hr, 1 hr)-temporal triangle counts over all eight temporal triangle types as shown in~\Fig{temp_tri}. (b):We fix $\delta_{1,3}$ to 2 hrs. We vary $\delta_{1,2}$ from 0 to 60 minutes and plot the ratio of (2 hrs, $\delta_{1,2},\delta_{2,3}$)-temporal triangles to (2hrs, $\delta_{1,2}$, 1hr)-temporal triangles for $\delta_{2,3}$ ranging from 0 to 60 minutes. (c) We plot the ratio of (2 hrs, 1 hr, $\delta_{2,3}$)-temporal triangles to (2hrs, 1hr, 1hr)-temporal triangles for $\delta_{2,3}$ ranging from 0 to 60 minutes, for cyclic and acyclic triangles.}
\label{fig:count_exp}
\end{figure*}

We implemented our algorithm in C++ and used a commodity machine from AWS EC2: R5d.2xlarge to run our experiments. This EC2 instance has Intel(R) Xeon(R) Platinum 8175M CPU @ 2.50GHz and 64GB memory. On this AWS machine, \PBL runs out of memory for the Bitcoin graph, so we used one with more than 256GB memory for this case. The implementation of \DOTTT{} is available at~\cite{ettc}.

We performed our experiments on a collection of temporal graphs from SNAP~\cite{snapnets}, KONECT\cite{konect}, and the Bitcoin transaction dataset from~\cite{kondor2014inferring}, consisting of all transactions up to Feb 9, 2018. The timestamp of each transaction is the creation time of the block on the blockchain that contains it\cite{reid2013analysis}.

{\bf Running time:}
All the running times are shown in~\Tab{runtime}. We ran all experiments on a single thread. In most instances, \DOTTT{} takes a few seconds to run. For graphs with tens of millions of temporal edges, \DOTTT{} runs in less than ten minutes. Even for the Bitcoin graph with 515M edges, \DOTTT{} takes less than an hour. 

{\bf Running time independent of time periods:}
The running time of both \DOTTT{} and \PBL algorithms are independent of the time periods. \DOTTT{} has the same running time for time restrictions ranging from 0 to the time span of the input dataset. For comparison with $\PBL$, we set $\delta_{1,3}=\delta_{1,2}=\delta_{2,3}=$ 1 hr.

{\bf Comparison with \PBL:}
We compare our algorithm with the \PBL algorithm that counts $\delta_{1,3}$-temporal triangles, as it is the closest to our work. We typically get a 1.5x-2x speedup over \PBL for large graphs (more than 0.5M edges) as shown in~\Fig{speedup}. Note that \DOTTT{} computes $\threeDelta$-temporal triangle counts while \PBL only gets the counts of $\delta_{1,3}$-temporal triangles.


{\bf Distribution of counts over types of triangles:}
The distribution of (1 hr, 1 hr, 1 hr)-temporal triangle counts for our datasets are shown in~\Fig{count_dist}. As we expected~\cite{milo2004superfamilies,vazquez2004topological,yaverouglu2014revealing,paranjape2017motifs}, networks from similar domains have similar distributions. It is easy to see in~\Fig{count_dist}, that all the stack exchange networks have similar distributions. The same holds for the message networks CollegeMsg and SMS-A.

We observe that cyclic temporal triangles, $\TriType_4$ and $\TriType_8$, have a larger share in temporal triangle counts in messaging networks than in stack exchange networks.


{\bf Triadic closures in temporal networks:}
In static triangles, the transitivity measures the ratio of number of static triangles to the number of all wedges. In temporal graphs, in addition to transitivity, the time it takes for a wedge to appear and close is of importance~\cite{zignani2014link}. In~\Fig{closure_delta1}, we study the effect of the time it takes for a wedge to appear from an edge, on the time it takes to close for CollegeMsg graph. We fix $\delta_{1,3} = $ 2 hrs. For $\delta_{1,2}$ ranging from zero to 60 minutes (10 minute steps), we vary $\delta_{2,3}$ from zero to 60 minutes and plot the ratio of $\threeDelta$-temporal triangles over (2hrs, $\delta_{1,2}$, 1hr)-temporal triangles. We observe that the set of ratios for all values of $\delta_{2,3}$ are almost identical for different values of $\delta_{1,2}$. For instance, for all values of $\delta_{1,2}$, roughly half the triangles are formed in 10-20 minutes. This implies that once a wedge is formed, the time it took to appear does not affect the time it takes to close.

As another demonstration of \DOTTT{}, for $\delta_{1,3} = $ 2 hrs and $\delta_{1,2}=$ 1 hrs, we plot the ratio of $\threeDelta$-temporal triangles to (2hrs, 1hr, 1hr)-temporal triangles, this time separately for cyclic and acyclic temporal triangles in~\Fig{closure_cyclic}. We observe that for stack exchange networks, acyclic temporal triangles tend to take a shorter time to close from the moment their second edge appears than cyclic temporal triangles. As we see in~\Fig{closure_cyclic}, this is not the case for message networks.

\begin{acks}
The authors are supported by NSF DMS-2023495, NSF CCF-1740850, NSF CCF-1813165, CCF-1909790, and ARO Award W911NF1910294. Noujan Pashanasangi is supported by Jack Baskin and Peggy Downes-Baskin Fellowship.
\end{acks}


\bibliographystyle{acm}
\bibliography{ref}

\newpage
\newpage
\clearpage

\appendix
\section{Missing Algorithms from section 5}\label{apx:missing_algo}
Here we will provide the algorithms for counting $\threeDelta$-temporal triangles for temporal orderings $\tempOrder_3$, $\tempOrder_4$, $\tempOrder_5$, and $\tempOrder_6$.

{\bf Temporal orderings $\tempOrder_3$ and $\tempOrder_4$:}
Consider an orientation $\dirType$, and a static triangle on vertices $\{u,v,w\}$ enumerated in $G_\prec$, where $u$ is the source vertex. The category of temporal orderings $\tempOrder_3$ and $\tempOrder_4$ is more intricate because $\tempOrder_3(\{v,w\})=\tempOrder_4(\{v,w\})=2$. Recall that we want to avoid enumerating temporal edges on $\{v,w\}$, so we do not enumerate edges on $S_2$ as in the case of $\tempOrder_1$ and $\tempOrder_2$. Instead, we enumerate edges on $S_1$, and compute the counts of edges on $S_2$ that form a temporal triangle with compatible edge in $S_3$.

We start by enumerating edges on $S_1$. Consider an edge $e_1 \in S_1$. Let $\ell_{\first}$ and $\ell_{\last}$ denote the indices of first and last edge in $S_3$ with a timestamp in $[t(e_1),t(e_1)+\delta_{1,3}]$. Also, let $\ell_{\delta_{1,2}}$ denote the index of the first edge in $S_3$ with a timestamp greater than $t(e_1) + \delta_{1,2}$, and $\ell_{\delta_{2,3}}$ be the index of the first edge in $S_3$ that has a timestamp greater than $t(e_1) + \delta_{2,3}$. Note that $\ell_{\delta_{1,2}}$ and $\ell_{\delta_{2,3}}$ divide $S_3[\ell:\ell_{\last}]$ into three consecutive subsequences. We show how to count temporal triangles that involve temporal edges in each of these three subsequences.

For a temporal edge $S_3[i]$ where $\ell_{\first} \leq i < \min(\ell_{\delta_{1,2}},\ell_{\delta_{2,3}})$, each edge $e_2 \in S_2$ where $t(e_2) \in [t(e_1),t(S_3[i])]$ form a $\threeDelta$-temporal triangle with $e_1$ and $S_3[i]$. To obtain the counts of these edges in $S_2$, it suffices to store $\cumulativeEdgeCount_{-\infty}$ on $S_2$ for each edge $e_3 \in S_3$.

Now, consider the temporal edges $S_3[i]$ where $\max(\ell_{\delta_{1,2}},\ell_{\delta_{2,3}}) < i \leq \ell_{\last}$. The timestamp of these edges are in time window $[t(e_1)+\max(\delta_{1,2},\delta_{2,3}),t(e_1)+\delta_{1,3}]$, and together with temporal edge $e_1$ form a $\threeDelta$-temporal triangle with each temporal edge $e_2 \in S_2$ where $t(e_2) \in [t(S_3[i])-\delta_{2,3},t(e_1)+\delta_{1,2}])$. To count the number of such temporal edges in $S_2$ we only need to store $\cumulativeEdgeCount_{-\delta_{2,3}}$ and $\cumulativeEdgeCount_{-\infty}$ on $S_2$ for each temporal edge $e_3$ in $S_3$.


Finally, consider a temporal edge $S_3[i]$ where $\min(\ell_{\delta_{1,2}},\ell_{\delta_{2,3}}) \leq i \leq \max(\ell_{\delta_1},\ell_{\delta_{2,3}})$. The number of compatible edges in $S_2$ depend on how $\delta_{1,2}$ compares to $\delta_{2,3}$. There are two cases: $(a): \delta_{1,2} < \delta_{2,3}$ and $(b): \delta_{2,3} < \delta_{1,2}$. In case $(a)$ edges in $S_2$ have to have a timestamp in $[t(e_1),t(e_1)+\delta_{1,2}]$ to form a triangle with $e_1$ and $S_3[i]$, and in case $(b)$ their timestamps should be in the time window $[t(S_3[i]) - \delta_{2,3}, t(S_3[i])]$.~\Alg{tempOrder_3_count} give the step by step procedure for counting temporal triangles for temporal orderings $\tempOrder_3$ and $\tempOrder_4$.


\begin{algorithm}[bh]
\caption{Counting $\threeDelta$-temporal triangles corresponding to a static triangle and temporal orientation $\tempOrder_3$ or $\tempOrder_4$}\label{alg:tempOrder_3_count}
\begin{algorithmic}[1]
\Procedure{CTT-vw2}{$\delta_{1,3}$, $\delta_{1,2}$, $\delta_{2,3}$,$\langle u,v,w \rangle$, $\tempOrder$, $\dirType$}
\LeftComment{$\tempOrder(\{v,w\})=2$}
\State count = 0
\State Enumerate $S_3$ and compute $\cumulativeEdgeCount_{-\infty}$ and $\cumulativeEdgeCount_{-\delta_{2,3}}$ on $S_2$
\For{$i= 1,\ldots,\sigma_1$}
\State Let $\ell_{\first} = \textsc{lowerBound}(t(S_1[i]), S_3)$
\State Let $\ell_{\delta_{1,2}} = \textsc{upperBound}(t(S_1[i])+\delta_{1,2}, S_3)$
\State Let $\ell_{\delta_{2,3}} = \textsc{lowerBound}(t(S_1[i])+\delta_{2,3}, S_3)$
\State Let $\ell_{\last}= \textsc{upperBound}(t(S_1[i])+\delta_{1,3}, S_3)$
\State Let $\ell_{\min} = \min(\ell_{\delta_{1,2}},\ell_{\delta_{2,3}})$
\State Let $\ell_{\max} = \max(\ell_{\delta_{1,2}},\ell_{\delta_{2,3}})$
\LeftComment Edges in $S_3[\ell_{\first}:\ell_{\min}]$
\State count $+= \cumulativeEdgeCount_{-\infty}(S_3[\ell_{\first}:\ell_{\min}],S_2)$
\State count $-= (\ell_{\min} - \ell_{\first} + 1)\cdot \edgeCount((-\infty,t(S_1[i])], S_2)$

\LeftComment Edges in $S_3[\ell_{\min}+1:\ell_{\max}-1]$

\If{$\delta_{1,2} \leq \delta_{2,3}$}
\State count $+=(\ell_{\delta_{2,3}} - \ell_{\delta_{1,2}})$
\State\hspace{\algorithmicindent} $\cdot \edgeCount([t(S_1[i]),t(S_1[i])+\delta_{1,2}], S_2)$
\ElsIf{$\delta_{2,3} \leq \delta_{1,2}$}
\State count $+= \cumulativeEdgeCount_{-\delta_{2,3}}(S_3[\ell_{\delta_{2,3}}:\ell_{\delta_{1,2}}],S_2)$
\EndIf

\LeftComment Edges in $S_3[\ell_{\max}:\ell_{\last}]$
\State count $+= \cumulativeEdgeCount_{-\delta_{2,3}}(S_3[\ell_{\max}:\ell_{\last}],S_2)$
\State count $-= \cumulativeEdgeCount_{-\infty}(S_3[\ell_{\max}:\ell_{\last}],S_2)$
\State count += $(\ell_{\last}-\ell_{\max})\cdot \edgeCount((-\infty,t(S_1[i])+\delta_{1,2}], S_2)$
\EndFor
\State \Return count
\EndProcedure
\end{algorithmic}
\end{algorithm}

{\bf Temporal orderings $\tempOrder_5$ and $\tempOrder_6$:} 
This case is similar to the case of temporal ordering $\tempOrder_1$ and $\tempOrder_2$. The difference is that while enumerating edges in $S_2$, we will first find the compatible edges in $S_3$ instead of $S_1$, and then count edges in $S_1$ that complete a temporal triangle. Consider a static triangle $\{u,v,w\}$ and orientation $\dirType$. Fix a temporal edge $e_2$ in $S_2$. Let $\ell_{\first}$ and $\ell_{\last}$ denote the indices of the first and last temporal edges in $S_3$ with a timestamp in the time period $[t(e_2),t(e_2)+\delta_{2,3}]$. Let $\ell_{\delta_{1,2}}$ be the first temporal edge in $S_3$ such that $t(S_3(\ell_{\delta_{1,2}})) > t(e_2) + \delta_{1,3} - \delta_{1,2}$. Since $\delta_{1,3} \leq \delta_{1,2} + \delta_{2,3}$, we have $\ell_{\first} \leq \ell_{\delta_{2,3}} \leq \ell_{\last}$. Consider a temporal edge $S_3[i]$ where $\ell_{\first} \leq i < \ell_{\delta_{1,2}}$. The number of edges in $S_1$ that form a triangle with $e_2$ and $S_3[i]$ is $\edgeCount(S_1,[t(e_2)-\delta_{1,2}, t(e_2)])$. For each temporal edge $S_3[i]$ where $\ell_{\delta_{1,2}} \leq i \leq \ell_{\last}$, there are $\edgeCount(S_1, [t(S_3[i])-\delta_{1,3}, t(e_2)])$ edges that complete a temporal triangle. This is the same as in the case of $\tempOrder_1$ and $\tempOrder_2$ with different time windows. We need to get $\cumulativeEdgeCount_{-\infty}$ and $\cumulativeEdgeCount_{-\delta_{1,3}}$ on $S_1$ for each edge $e_3 \in S_3$.

\begin{algorithm}[bh]
\caption{Counting $\threeDelta$-temporal triangles corresponding to a static triangle and temporal orientation $\tempOrder_5$ or $\tempOrder_6$}\label{alg:tempOrder_5_count}
\begin{algorithmic}[1]
\Procedure{TTC-vw1}{$\delta_{1,3}$, $\delta_{1,2}$, $\delta_{2,3}$,$\langle u,v,w \rangle$, $\tempOrder$, $\dirType$}
\LeftComment{$\tempOrder(\{v,w\})=1$}
\State count = 0
\State Enumerate $S_3$ and compute $\cumulativeEdgeCount_{-\delta_{1,3}}$ and $\cumulativeEdgeCount_{-\infty}$ on $S_1$
\For{$i= 1,\ldots,\sigma_2$}
\State Let $\ell_{\first} = \textsc{lowerBound}(t(S_2[i]), S_3)$
\State Let $\ell_{\delta_{1,2}} = \textsc{lowerBound}(t(S_2[i])+\delta_{1,3}-\delta_{1,2}, S_3)$
\State Let $\ell_{\last}= \textsc{upperBound}(t(S_2[i])+\delta_{2,3}, S_3)$
\LeftComment Edges in $S_3[\ell_{\delta_{1,2}}:\ell_{\last}]$
\State count $+= \cumulativeEdgeCount_{-\delta_{1,3}}(S_3[\ell_{\delta_{1,2}}:\ell_{\last}],S_1)$
\State count $-= \cumulativeEdgeCount_{-\infty}(S_3[\ell_{\delta_{1,2}}:\ell_{\last}],S_1)$
\State count $+= (\ell_{\last}-\ell_{\delta_{1,2}}+1) \cdot \edgeCount((-\infty,t(S_2[i])], S_1)$
\LeftComment Edges in $S_3[\ell_{\first}:\ell_{\delta_{1,2}}-1]$
\State count $+= (\ell_{\delta_{1,2}}-\ell_{\first}) \cdot \edgeCount([t(S_2[i])-\delta_{1,2}, t(S_2[i])], S_1)$
\EndFor
\State \Return count
\EndProcedure
\end{algorithmic}
\end{algorithm}

\end{document}